\documentclass[journal]{IEEEtran}

\usepackage{cite}
\usepackage{amsthm}
\usepackage{bm}
\usepackage{amsmath,amssymb,amsfonts}
\usepackage{mathtools}
\usepackage{algorithm,algorithmic}
\usepackage{makecell}
\usepackage{array}
\usepackage{etoolbox}
\usepackage{blindtext}
\usepackage{stfloats} 
\usepackage{scalerel,stackengine}
\DeclareMathOperator{\vect}{vec}
\usepackage{mathrsfs}
\usepackage{listings}
\usepackage{booktabs}
\usepackage[acronym,shortcuts]{glossaries}

\usepackage[caption=false,font=footnotesize]{subfig}
\usepackage[dvipsnames]{xcolor}
\newcommand{\ba}{\begin{array}}
\newcommand{\ea}{\end{array}}
\newcommand{\be}{\begin{displaymath}}
\newcommand{\ee}{\end{displaymath}}
\newcommand{\ben}{\begin{equation}}
\newcommand{\een}{\end{equation}}
\newcommand{\bena}{\begin{eqnarray}}
\newcommand{\eena}{\end{eqnarray}}
\newcommand{\beqa}{\begin{eqnarray*}}
\newcommand{\enqa}{\end{eqnarray*}}

\newcommand{\bc}{\begin{center}}
\newcommand{\ec}{\end{center}}
\newcommand{\bi}{\begin{itemize}}
\newcommand{\ei}{\end{itemize}}
\newcommand{\benu}{\begin{enumerate}}
\newcommand{\eenu}{\end{enumerate}}
\newcommand{\bdes}{\begin{description}}
\newcommand{\edes}{\end{description}}
\newcommand{\bt}{\begin{tabular}}
\newcommand{\et}{\end{tabular}}

\newcommand \zetabf{\boldsymbol{\zeta}}

\newcommand \mubf{\boldsymbol{\mu}}

\newcommand \chibf{\boldsymbol{\chi}}

\newcommand \abf{{\bf a}}

\newcommand \dbf{{\bf d}}

\newcommand \gbf{{\bf g}}

\newcommand \nbf{{\bf n}}

\newcommand \qbf{{\bf q}}

\newcommand \vbf{{\bf v}}
\newcommand \wbf{{\bf w}}
\newcommand \xbf{{\bf x}}
\newcommand \ybf{{\bf y}}

\newcommand \Abf{{\bf A}}
\newcommand \Bbf{{\bf B}}
\newcommand \Cbf{{\bf C}}
\newcommand \Dbf{{\bf D}}
\newcommand \Ebf{{\bf E}}
\newcommand \Fbf{{\bf F}}
\newcommand \Gbf{{\bf G}}
\newcommand \Hbf{{\bf H}}
\newcommand \Ibf{{\bf I}}

\newcommand \Mbf{{\bf M}}
\newcommand \Nbf{{\bf N}}

\newcommand \Qbf{{\bf Q}}
\newcommand \Rbf{{\bf R}}

\newcommand \Wbf{{\bf W}}
\newcommand \Xbf{{\bf X}}
\newcommand \Ybf{{\bf Y}}

\newcommand{\circlambda}{\mbox{$\Lambda$
             \kern-.85em\raise1.5ex
             \hbox{$\scriptstyle{\circ}$}}\,}

\newcommand{\tr}{\mathop{\rm tr}}

\newtheorem{lemma}{Lemma}%[section]
\newtheorem{remark}{Remark}%[section]

\makeatletter
\newcommand{\ALOOP}[1]{\ALC@it\algorithmicloop\ #1%
	\begin{ALC@loop}}
	\newcommand{\ENDALOOP}{\end{ALC@loop}\ALC@it\algorithmicendloop}

\makeatother



\usepackage{scalerel}
\usepackage{tikz}
\usetikzlibrary{svg.path}

\definecolor{orcidlogocol}{HTML}{A6CE39}
\tikzset{
  orcidlogo/.pic={
    \fill[orcidlogocol] svg{M256,128c0,70.7-57.3,128-128,128C57.3,256,0,198.7,0,128C0,57.3,57.3,0,128,0C198.7,0,256,57.3,256,128z};
    \fill[white] svg{M86.3,186.2H70.9V79.1h15.4v48.4V186.2z}
                 svg{M108.9,79.1h41.6c39.6,0,57,28.3,57,53.6c0,27.5-21.5,53.6-56.8,53.6h-41.8V79.1z M124.3,172.4h24.5c34.9,0,42.9-26.5,42.9-39.7c0-21.5-13.7-39.7-43.7-39.7h-23.7V172.4z}
                 svg{M88.7,56.8c0,5.5-4.5,10.1-10.1,10.1c-5.6,0-10.1-4.6-10.1-10.1c0-5.6,4.5-10.1,10.1-10.1C84.2,46.7,88.7,51.3,88.7,56.8z};
  }
}

\newcommand\orcidicon[1]{\href{https://orcid.org/#1}{\mbox{\scalerel*{
\begin{tikzpicture}[yscale=-1,transform shape]
\pic{orcidlogo};
\end{tikzpicture}
}{|}}}}
\definecolor{darkgreen}{RGB}{0,100,0}    
\definecolor{darkblue}{RGB}{0,0,139}    
\definecolor{darkred}{rgb}{0.55, 0.0, 0.0}
\usepackage[colorlinks=true,
            urlcolor=blue,
            citecolor=red,
            linkcolor=blue]{hyperref}

\makeatletter
\expandafter\let\expandafter\IEEEbiographyinternal
  \csname\string\IEEEbiography\endcsname

\patchcmd{\IEEEbiographyinternal}
  {\vskip \@IEEEBIOskipN plus 1fil minus 0\baselineskip}
  {\vskip 2\baselineskip plus 0.25\baselineskip minus 0.1\baselineskip}
  {}
  {\PackageWarning{main}{Could not patch IEEEbiography spacing}}

\expandafter\let
  \csname\string\IEEEbiography\endcsname
  \IEEEbiographyinternal
\makeatother
\newcounter{editnum}

\newacronym{ISAC}{ISAC}{Integrated sensing and communication}
\newacronym{6G}{6G}{sixth-generation}
\newacronym{NOMA}{NOMA}{Non-Orthogonal Multiple Access}
\newacronym{BS}{BS}{base station}
\newacronym{Eve}{Eve}{eavesdropper}
\newacronym{AO}{AO}{alternating optimization}
\newacronym{SDR}{SDR}{semidefinite relaxation}
\newacronym{SCA}{SCA}{successive convex approximation}
\newacronym{PD-NOMA}{PD-NOMA}{power-domain non-orthogonal multiple access}
\newacronym{SIC}{SIC}{interference cancellation}
\newacronym{CSI}{CSI}{channel state information}
\newacronym{SINR}{SINR}{signal-to-interference-plus-noise ratio}
\newacronym{TDMA}{TDMA}{time-division-multiple-access}
\newacronym{UAV}{UAV}{uncrewed aerial vehicle}
\newacronym{CRB}{CRB}{Cramér–Rao bound}
\newacronym{PLS}{PLS}{Physical layer security}
\newacronym{MSE}{MSE}{mean squared error}
\newacronym{V2X}{V2X}{vehicle-to-everything}
\newacronym{RIS}{RIS}{reconfigurable intelligent surface}
\newacronym{RCS}{RCS}{radar cross section}
\newacronym{AWGN}{AWGN}{additive white Gaussian noise}
\newacronym{FIM}{FIM}{Fisher Information Matrix}
\newacronym{ULAs}{ULAs}{uniform linear arrays}
\newacronym{AoA}{AoA}{angle-of-arrival}
\newacronym{AoD}{AoD}{angle-of-departure}
\newacronym{OP}{OP}{optimization problem}
\newacronym{CPI}{CPI}{coherent processing interval}
\newacronym{CS}{C\&S}{communication and sensing}
\newacronym{MF}{MF}{matched-filter}
\newacronym{SI}{SI}{self-interference}
\newacronym{DoF}{DoF}{degree of freedom}
\newacronym{AO-SCA}{AO-SCA}{alternating optimization–successive convex approximation}

\begin{document}
\title{Uncertainty-Aware Beamforming and Power Allocation for Uplink NOMA--ISAC}

\author{
Azadeh Tabeshnezhad\textsuperscript{\orcidicon{0000-0003-0140-0076}},~\IEEEmembership{Member,~IEEE}, 
Milad Tatar Mamaghani\textsuperscript{\orcidicon{0000-0002-3953-7230}},
A. Lee Swindlehurst\textsuperscript{\orcidicon{0000-0002-0521-3107}},~\IEEEmembership{Life Fellow,~IEEE},
Tommy Svensson\textsuperscript{\orcidicon{0000-0002-2579-9002}},~\IEEEmembership{Senior Member,~IEEE}, and
Erik Ström\textsuperscript{\orcidicon{0000-0002-3084-7232}},~\IEEEmembership{Fellow,~IEEE}
\thanks{This work was supported in part by the ROBUST-6G project through the Smart Networks and Services Joint Undertaking (SNS JU) under the European Union's Horizon Europe research and innovation programme, Grant Agreement No. 101139068, and in part by the U.S. National Science Foundation under grant CCF-2322191.}
\thanks{Azadeh Tabeshnezhad, Tommy Svensson, and Erik Ström are with the Department of Electrical Engineering, Chalmers University of Technology, 412 96 Gothenburg, Sweden (e-mail: azadeh.tabeshnezhad@chalmers.se; tommy.svensson@chalmers.se; erik.strom@chalmers.se).}
\thanks{M. Tatar Mamaghani is an independent researcher in Sydney, NSW, Australia (e-mail: milad.tatarmamaghani@gmail.com).}
\thanks{A. Lee Swindlehurst is with the Department of Electrical Engineering \& Computer Science, University of California, Irvine, CA 92697, USA (e-mail: \href{mailto:swindle@uci.edu}{\textcolor{black}{swindle@uci.edu}}).}
}

\maketitle

%% Abstract must be <=200 words 
\begin{abstract}
In this paper, we investigate secure uplink \ac{PD-NOMA} in an integrated sensing and communication (ISAC) system where a full-duplex \ac{BS} simultaneously receives uplink transmissions and transmits a sensing waveform to estimate and jam an eavesdropping target (Eve) whose location is imperfectly known. To explicitly account for sensing-derived uncertainty, we use the angle--delay submatrix of the Cramér--Rao bound as an error-covariance proxy and propagate the BS--Eve estimation errors to the user--Eve link parameters through geometry-dependent Jacobians. This approach enables tractable uncertainty-aware approximations of the sensing and information-leakage metrics. We then jointly optimize the BS receive combiners, user transmit beamformers, and sensing power to maximize a weighted communication--sensing utility subject to transmit-power constraints and a statistical leakage constraint based on uncertainty-averaged received powers. The resulting nonconvex problem is addressed through alternating optimization: the receive combiners have closed-form generalized-eigenvector updates, while the transmit beamformers and sensing power are optimized using a quadratic transform and \ac{SDR}. Numerical results show that the proposed design maintains Eve's statistical SINR below the required threshold under location uncertainty while achieving strong sensing performance and communication rate, with only a small secrecy-rate loss relative to an ideal secure design with perfect Eve-location knowledge.
\end{abstract}

% Note that keywords are not normally used for peer review papers.

\IEEEpeerreviewmaketitle

%\listoftodos
\section{Introduction}

\IEEEPARstart{T}{he} next generation of mobile communication networks is expected to support a diverse range of mission-critical applications, including autonomous systems, industrial automation, and connected smart societal infrastructure. These emerging services demand not only ultra-reliable and low-latency communications but also high-precision environmental awareness and stringent digital security. \ac{ISAC} has emerged as a promising technology for \ac{6G} networks by enabling the joint operation of radar sensing and data communication through the use of shared spectrum and hardware resources ~\cite{EJ2021Zhang, SFL2022Liu} and via unified waveform design \cite{MRC2020Mert}. Thus, ISAC has the potential to reduce hardware cost, improve spectral and energy efficiency, and facilitate real-time situational awareness in 6G mobile networks. However, ISAC also introduces new design challenges, particularly in managing the interference between communication and sensing signals, and safeguarding such systems against malicious actors.

Concurrently, \ac{PD-NOMA} has emerged as a compelling solution for multi-user access. NOMA improves spectral efficiency through superposition coding at the transmitter and successive \ac{SIC} at the receiver~\cite{RAI2023Azadeh, JMAR2024Azadeh}, distinguishing multiple users by their signal power and thus enabling them to simultaneously occupy the same time, frequency, and spatial resources. NOMA also opens new opportunities for interference-resilient and resource-constrained ISAC system design by providing new degrees of freedom for designing robust and secure waveforms against eavesdroppers in shared-spectrum environments.

\subsection{Related work}

%% NOMA ISAC-related works with no security aspects
Recent advances in ISAC systems have leveraged NOMA to improve spectral efficiency and manage interference. For instance, \cite{OSN2024Sun} optimizes resource sharing using outage and \ac{MSE} metrics for a NOMA--ISAC system, while \cite{NEI2022Wang} proposes beamforming designs to maximize a joint communication-sensing utility.
Other efforts, such as \cite{ISC2023Amhaz} and \cite{SEN2023Dou}, extend NOMA--ISAC by jointly designing transmission duration, beamforming, and power control to enhance sensing efficiency and user rates in terrestrial and satellite scenarios. Despite their promising results, these studies focus on the downlink, assume ideal \ac{CSI}, and do not address performance robustness under location uncertainty. The authors of \cite{NII2022Wang} investigate a NOMA-inspired ISAC system with multiple users and targets, in which a dual-functional BS jointly optimizes communication and sensing beams while enforcing rate and \ac{SIC} constraints. 

To improve sensing coverage and beam pattern gain, \cite{NAF2024Xue} employs a \ac{RIS}. The sensing \ac{SINR} for a full-duplex NOMA--ISAC system is optimized in \cite{BOF2022Chen}. Other research such as \cite {SIS2023Zhang} and \cite{PNS2023Akhtar} divides the spectrum between dedicated joint sensing-communication modes, offering flexibility under certain conditions. While the studies above improve NOMA--ISAC performance through novel architectures and designs, they often neglect the security challenges of ISAC systems, particularly from the physical-layer perspective.

% 3) NOMA--ISAC works with Secrecy 
\ac{PLS}, which utilizes local channel state information and signaling to protect communications, has recently attracted significant interest in the context of ISAC system design \cite{SPO2022Yang, SPS2024Huang, SIS2025Milad}. Yang \emph{et al.} in \cite{SPO2022Yang} present a secure NOMA-enabled multiuser ISAC framework in which joint precoding and artificial jamming are optimized via \ac{SCA} to maximize the secrecy rate while preserving sensing performance. The authors of \cite{SPS2024Huang} design secure precoding schemes for satellite NOMA–ISAC systems to enhance secrecy while preserving sensing performance. Their results show that the proposed designs significantly enhance secrecy performance compared to \ac{TDMA} transmission while maintaining sensing reliability.  Mamaghani \emph{et al.} in \cite{SIS2025Milad} model the interaction between an ISAC \ac{BS} and a mobile malicious \ac{UAV} as a Stackelberg game, combining SCA-based resource optimization with deep reinforcement learning for adaptive adversarial trajectory control, achieving improved secrecy, sensing accuracy, and power efficiency, while preserving robust performance against a proactive mobile threat.

Despite these advances, prior security-focused ISAC work predominantly assumes perfect eavesdropper CSI and does not model how sensing-parameter estimation errors propagate through the ISAC processing. Although estimation-theoretic metrics such as the \ac{CRB} have been incorporated into ISAC design, e.g., CRB minimization for NOMA beamforming \cite{JUP2023Dou} and performance characterization of uplink ISAC with spectrum sharing \cite{OPU2022Ouyang}, these approaches use the CRB primarily as an isolated sensing metric. They do not integrate CRB-based uncertainty into robust beamforming design, nor do they analyze how estimation inaccuracies affect communication reliability or secrecy. The fundamental coupling between CRB-driven sensing errors and system-level performance metrics such as SINR, data rate, and secrecy capacity remains insufficiently understood, motivating the current research.

\subsection{Our Contributions}
We consider a novel scenario in which an ISAC system employing uplink \ac{PD-NOMA} must also provide security against an eavesdropping target with an imprecisely known location. Unlike prior work that considers imperfect CSI or sensing-parameter estimation, we incorporate knowledge of the estimation-theoretic uncertainty into the beamforming design. In particular, we introduce a CRB-informed design framework to capture the impact of parameter estimation errors on communication, sensing, and secrecy performance. Our detailed contributions are as follows:

\begin{itemize}
    \item We develop an uncertainty-aware model for an uplink NOMA--ISAC system with an eavesdropping target, and characterize how sensing-estimation errors for the target's angle and delay parameters propagate and affect the sensing and system security. We then derive tractable SINR approximations for the communication, sensing, and eavesdropping performance metrics.
    \item We formulate a joint beamforming and sensing-power optimization problem that maximizes a weighted communication--sensing utility while enforcing security and power constraints, and solve it via an alternating-optimization algorithm combining closed-form receive-combiner updates with \ac{SDR} based transmit optimization, which yields a computationally efficient iterative solution. 
    \item We demonstrate via simulations that the proposed robust-secure design keeps the eavesdropper SINR below a prescribed threshold under location uncertainty while preserving adequate sensing and communication performance relative to ideal and non-secure baselines.
\end{itemize}

\subsection{Organization and Notation}
The rest of the paper is organized as follows. Section~\ref{sec:SystemModel} presents the uplink \ac{PD-NOMA} ISAC system model. Section~\ref{sec:metrics} introduces the communication, sensing, and eavesdropping performance metrics. Section~\ref{sec:CRB} develops a CRB-based uncertainty analysis and derives the resulting uncertainty-aware SINR expressions for communications and sensing. Section~\ref{sec:Joint_Opt} formulates the joint optimization problem and presents the proposed \ac{AO-SCA} solution. Section~\ref{sec:simulations} provides numerical results, followed by conclusions drawn in Section~\ref{sec:Conclusion}.

\textit{Notation}: Lowercase letters denote a scalar, lowercase boldface letters denote a vector, and uppercase boldface letters denote a matrix. The element of matrix $\mathbf{H}$ in the $i$-th row and $j$-th column is denoted as $(\mathbf{H})_{i,j}$. The operators $(\cdot)^\top$, $(\cdot)^*$, and $(\cdot)^H$ denote transpose, complex conjugation, and Hermitian transpose, respectively. The Euclidean norm of a vector $\mathbf{x}$ is given by $\|\mathbf{x}\|$, $\Re\{.\}$ denotes the real part of a complex quantity, and $\tr(\cdot)$ denotes the matrix trace. The notation $[\cdot]^+$ denotes the positive-part operator. The notation $\mathbf{I}_N$ represents the $N\times N$ identity matrix, and $\mathbf{1}_N$ denotes an $N\times 1$ vector of ones. The space of $p\times q$ complex-valued matrices is written as $\mathbb{C}^{p\times q}$. A circularly symmetric complex Gaussian multivariate random distribution with mean vector $\pmb{\mu}$ and covariance matrix  $\mathbf{C}$ is denoted as $\mathcal{CN}(\pmb{\mu}, \mathbf{C})$.

\section{System Model}\label{sec:SystemModel}
We consider an ISAC system employing uplink PD-NOMA, as shown in  Fig.~\ref{fig1}. A monostatic ISAC BS with well-separated transmit and receive arrays of $N_t$ and $N_r$ antennas, respectively, transmits a signal to simultaneously sense and jam an eavesdropping target (Eve), and receives uplink signals from two multi-antenna users (UEs) as well as the target echo signal and radar \ac{SI}. We assume that the transmit and receive arrays have the same number of antennas, i.e., $(N_t=N_r=N)$. The BS uses the target echo to estimate Eve's location, albeit with some estimation error. Each UE has $M$ antennas, while Eve is assumed to have a single antenna. Eve attempts to overhear the uplink transmissions in the presence of the BS radar interference, while the UEs employ beamforming to steer their signals away from Eve using the imprecise location information derived at the BS. We assume that the BS, UEs, and Eve are static\footnote{This assumption can be relaxed assuming the UEs and Eve are slowly moving compared to the involved time scales in the estimation and communications processes.} and located at coordinates $\dbf_b = [0,0,0]^\top$, $\dbf_k = [x_k,y_k,0]^\top,~k = \{1,2\}$, and $\dbf_e = [x_e, y_e, z_e]^\top$, respectively. For the subsequent three-dimensional location analysis, the azimuthal direction of Eve relative to the BS is assumed known.

\begin{figure}
    \centering
    \includegraphics[width=0.85\linewidth]{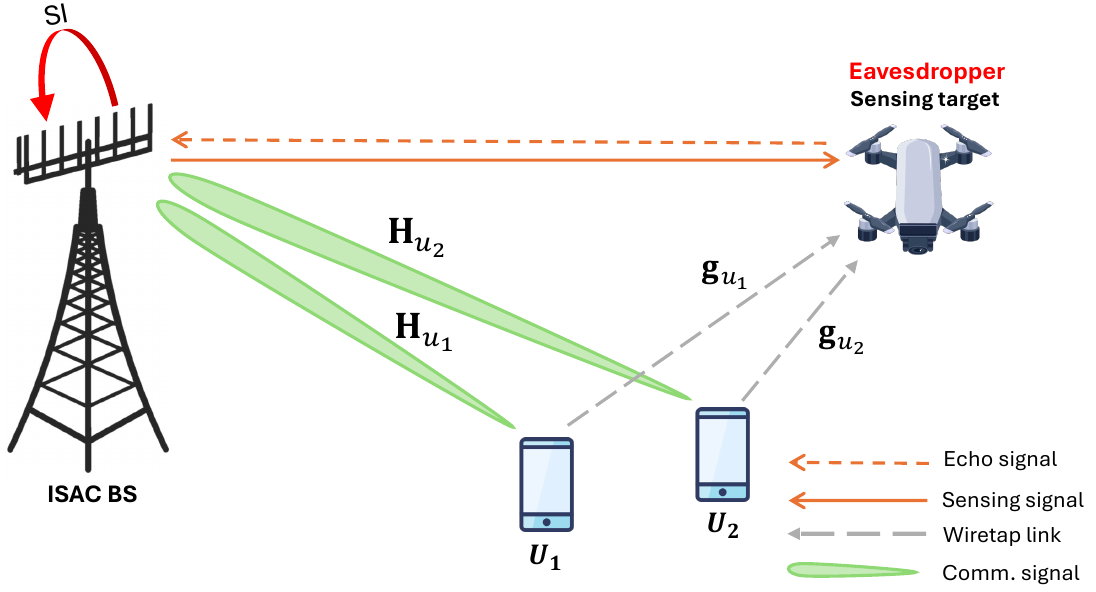}
    \caption{Illustration of a monostatic ISAC system employing PD-NOMA in the uplink and sensing an eavesdropping target in the downlink.}
    \label{fig1}
\end{figure}

\subsection{Signal Models}
At time $t$, the BS transmits the sensing signal
\begin{align}
  \xbf(t) & = \sqrt{p_s}\abf(\hat\theta_b)\sum_{m=0}^{L-1} a_m \phi\!\left(t-m T\right), \label{eu_radar:sig}
\end{align}
where $p_s$ is the transmit sensing power, $\abf(\hat\theta_b)$ is the array steering vector at the BS based on the estimated angle $\hat\theta_b$ to Eve, $a_m$ is the $m$-th symbol in the transmitted symbol sequence, $\phi(t)$ is a unit-energy complex baseband waveform, i.e., $\int_{-\infty}^{\infty} \big|\phi(t)\big|^2 \,dt=1$, $L$ is the number of symbol pulses of duration $T$. Simultaneously, the BS receives $\ybf_\text{b}(t) \in \mathbb{C}^{N_r\times1}$, which includes the radar echo from the target Eve, the uplink signals from the UEs, \ac{SI} due to $\xbf(t)$, and noise:
\begin{equation} \label{eu_yb}
\ybf_\text{b}(t) = \underbrace{\ybf_{\text{c}}(t)}_{\text{Communication}} \mathrel{+} 
\underbrace{\ybf_{\text{s}}(t)}_{\text{Sensing}} \mathrel{+} 
\underbrace{\ybf_{\rm SI}(t)}_{\text{\ac{SI}}}\mathrel{+}
\underbrace{\nbf_\text{b}(t)}_\text{noise}.
\end{equation} 
The \ac{SI} signal is given by
\[{\ybf_{\rm SI}(t)} = \Hbf_{\rm SI}\xbf(t), \] 
where the elements of $\Hbf_{\rm SI} \in \mathbb{C}^{N_r\times N_t}$ are given by \cite{FDM2024SGP}
\begin{align} \label{eu_SI}
(\Hbf_{\rm SI})_{i',j'} = \frac{\rho}{d_{i',j'}} \exp\left( -2\pi j \frac{d_{i',j'}}{\lambda} \right),
\end{align}
where $0\leq \rho\leq 1$ indicates the amount of SI cancellation prior to baseband processing ($\rho=0$ indicates perfect cancellation), and $d_{i',j'}$ is the distance between the $i'$-th receive and $j'$-th transmit antennas at the BS. The noise term in \eqref{eu_yb} is modeled as \ac{AWGN}: $\nbf_\text{b}(t)\sim \mathcal{CN} (\mathbf{0},{\sigma}^2_b \Ibf)$, where $\Ibf$ is the $N_r\times N_r$ identity matrix.

In the following, we provide expressions for the communication and sensing components in \eqref{eu_yb}, often dropping the time argument $t$ for brevity.

\subsubsection{\texorpdfstring{Communication signal $\ybf_c$}{Communication signal yc}}
The noise-free uplink communication signal $\ybf_c$ is given by
\begin{align} \label{eu_yc}
    \ybf_{\text{c}}(t) = \sum_{k=1}^2 \Hbf_{u_k} \widetilde \wbf_k s_k (t),
\end{align}
where $\Hbf_{u_k} \in \mathbb{C}^{N_r\times M}$ is the Rayleigh fading channel between the UE $k$ and BS, $s_{k}(t)$ is the unit-power baseband signal from UE $k$, and $\widetilde \wbf_k = \sqrt{p_k}\wbf_k$, where $\wbf_k \in \mathbb{C}^{M \times 1}$ is the UE transmit beamformer with power $p_k$ assuming $\|\wbf_k\|=1$.
 
\subsubsection{\texorpdfstring{Sensing signal $\ybf_s$}{Sensing signal ys}}
We assume that the sensing channel is a line-of-sight (LoS) link\footnote{Local scattering could be incorporated through an additional stochastic NLoS component; extending the uncertainty propagation and robust design to such composite channels is left for future work.}, e.g., as is representative of an elevated and largely unobstructed eavesdropping target such as higher-altitude UAVs \cite{MiladSSPC2024}. Hence, the noise-free sensing echo signal $\ybf_s$ as received at the BS can be expressed as
\begin{align} \label{eu_ys}
   \ybf_\text{s}(t) = \beta_0 \Abf(\theta_b) \xbf (t-2\tau_b),
   \end{align}  
where $\beta_0$ denotes the \ac{RCS} of the target, and $2\tau_b$ denotes the round-trip delay, where $\tau_b = \|\dbf_b - \dbf_e\|/c$ is the one-way delay and $c$ is the speed of light.
The sensing channel matrix $\Abf(\theta_b)$ is  expressed as
\begin{align}\label{eu_A(theta)}
   \Abf(\theta_b) = \alpha_2(\tau_b) \abf_r(\theta_b) \abf_t(\theta_b)^H,
\end{align}
where $\alpha_2(\tau_b)$ indicates the round-trip large-scale path gain between the BS and Eve, given by \cite{EWRS2013}
\begin{align}
   \vert \alpha_2(\tau_b) \vert^2 = \frac{\lambda^2}{(4\pi)^3( c\tau_b)^4},
\end{align}
where $\lambda = c/f_c$ is the wavelength and $f_c$ the carrier frequency. The vectors $\abf_r(\theta_b)$ and $\abf_t(\theta_b)$ respectively represent the receive and transmit array responses, which we assume are both \ac{ULAs} with identical orientations\footnote{All ULAs are assumed to be vertically oriented, so that their steering responses depend only on the elevation angle. The extension to arrays with two-dimensional azimuth--elevation responses is left as future work.} and separation that is large enough to avoid excessive \ac{SI}, but small enough so that the \ac{AoA} and \ac{AoD} are both given by $\theta_b$.
Thus, 
\begin{align} \label{eu_Steering}
    \abf_i(\theta_b) = \frac{1}{\sqrt{N_i}}[ 1, e^{j \frac{2\pi \delta}{\lambda}{\sin \theta_b}}, \dots, e^{j \frac{2\pi \delta}{\lambda}{(N_i-1) \sin \theta_b}} ]^T, 
  \end{align}  
where $\delta$ is the antenna spacing and $i\in\{r,t\}$.

\subsubsection{\texorpdfstring{Signal at Eve}{Signal at Eve}}
Eve receives a superposition of the UE uplink and the radar transmissions, which, under the assumption of LoS channels, can be expressed as
\begin{align} \label{eu_ye}
 y_e(t) &= \sum_{k=1}^2 \alpha_1(\tau_{u_k}) \abf(\theta_{u_k})^H \widetilde{\wbf}_k s_k (t) \nonumber\\
 & \qquad + \alpha_1(\tau_b) \abf_t(\theta_b)^H \xbf(t- \tau_b) + n_e(t),
\end{align}
where $\theta_{u_k}$ is the \ac{AoD} of the signal from UE$_k$ to Eve, $\abf(\theta_{u_k})$ is defined similar to \eqref{eu_Steering}, $n_e\sim \mathcal{CN}(0,\sigma^2_e)$, and
\begin{align}
    \vert \alpha_1(\tau_{u_k}) \vert^2 &= \left(\frac{\lambda}{4\pi c \tau_{u_k}}\right)^2. \label{eu:alpha_u} \\
    \vert \alpha_1(\tau_b) \vert^2 &= \left(\frac{\lambda}{4\pi c \tau_b}\right)^2, \label{eu:alpha_b}
\end{align}
are respectively the one-way large-scale path gains between Eve-UE $k$ and Eve-BS, with $\tau_{u_k} = \|\dbf_{k} - \dbf_e\|/c$.

\begin{remark}
Note that Eve is assumed to know the pulse shape and signaling protocol but cannot perfectly reconstruct and coherently subtract the realization of the BS sensing waveform over the current \ac{CPI}. This assumption is applicable when the pulse-modulation sequence is unavailable to Eve or when Eve lacks sufficiently accurate instantaneous knowledge of the BS--Eve channel. If both the complete sensing waveform and the corresponding channel were perfectly known to Eve, this interference could be partially or fully canceled; that stronger threat model is beyond the scope of this work.
\end{remark}

\section{Performance metrics}\label{sec:metrics}
\subsection{Communication Performance}
We adopt the \ac{SINR} as the metric to evaluate communication performance in the considered \ac{PD-NOMA} ISAC system. We assume that the BS has perfect CSI for both UEs, and that the first user (UE$_1$) is closer to the BS, yielding a favorable channel gain, i.e., \(\|\mathbf{H}_{u_1}\| \geq \|\mathbf{H}_{u_2}\|\). Following the NOMA decoding principle, the BS first decodes the communication signal from UE$_1$, treating UE$_2$’s signal as the main source of interference. Subsequently, the BS subtracts UE$_1$'s signal from $\ybf_b(t)$ via SIC, and then proceeds to decode UE$_2$'s signal. 

To this end, we assume that the BS employs receive combiners $\mathbf{v}_1, \mathbf{v}_2\in \mathbb{C}^{N_r\times 1}$ to recover the corresponding signals for UE$_1$ and UE$_2$, respectively. Accordingly, the SINR for decoding UE$_1$'s signal, treating the SI and UE$_2$'s signal as additional interference and ignoring the weak received reflected radar echo from target Eve, is given by
\begin{align} \label{eu_User1}
   \gamma_1 =  \frac{\vert \vbf^H_1 \Hbf_{u_1} \widetilde \wbf_1 \vert^2}{ \vert \vbf^H_1 \Hbf_{u_2} \widetilde \wbf_2 \vert^2 +\mathbb{E}\{\vert \vbf_1^H \Hbf_{\rm SI}\xbf \vert^2\}+\sigma^2_b},
\end{align}
where the expectation is taken over the sensing signal $\xbf$. Assuming that the SI can be sufficiently suppressed, the BS cancels UE$_1$'s signal without any error propagation. Thus, the SINR to decode UE$_2$’s transmission becomes

\begin{align} \label{eu_User2}
\gamma_2 = \frac{\vert \vbf^H_2 \Hbf_{u_2} \widetilde \wbf_2 \vert^2}{\mathbb{E}\{\vert\vbf_2^H \Hbf_{\rm SI}\xbf \vert^2\}+\sigma^2_b}.
\end{align}

\begin{remark}
   ISAC systems exhibit degraded performance compared to their communication-only or sensing-only counterparts; thus, the mutual interference between the radar and communication signals should be accounted for in the SINR evaluations. Nevertheless, the sensing echo signal is typically weak compared to the communication component, and hence the echo term $\mathbb{E}\{\vert \beta_0 \vbf_k^H \Abf(\theta_b) \xbf \vert^2\}$ is ignored in the SINR expressions of \eqref{eu_User1} and \eqref{eu_User2}.
\end{remark}

Using the definition in~\eqref{eu_radar:sig}, the residual \ac{SI} covariance due to the sensing beam depends on the estimated steering direction $\hat\theta_b$ and is given by
\begin{align}
    \Qbf_{\rm SI} \triangleq \Hbf_{\rm SI} \abf(\hat\theta_b)\abf(\hat\theta_b)^H \Hbf_{\rm SI}^H.
\end{align}
Thus, the SINRs in \eqref{eu_User1} and \eqref{eu_User2} can be approximated as
\begin{align} \label{eu_SINR_user1_xN}
    \boxed{\gamma_1 \approx  \frac{\vert \vbf^H_1 \Hbf_{u_1} \widetilde \wbf_1 \vert^2}
      {\vert \vbf^H_1 \Hbf_{u_2} \widetilde \wbf_2 \vert^2 +  p_s (\vbf_1^H \Qbf_{\rm SI}\vbf_1 ) +\sigma^2_b}.}
\end{align}
\begin{align} \label{eu_SINR_user2_xN}
    \boxed{\gamma_2 \approx \frac{\vert \vbf^H_2 \Hbf_{u_2} \widetilde \wbf_2 \vert^2}{p_s( \vbf_2^H \Qbf_{\rm SI}\vbf_2 )  +\sigma^2_b},}
\end{align}

\subsection{Sensing performance}
Assuming the UE messages are decoded at the BS and removed from $\ybf_\text{b}(t)$ in \eqref{eu_yb}, the signal used for sensing becomes
\begin{align} \label{eu:y_bs}
 \ybf_\text{b,s}(t) = \ybf_\text{s}(t) + \ybf_\text{SI}(t) + \nbf_\text{b}(t).
\end{align}
A \ac{MF} is applied over the $L$ transmitted pulses to reduce the effects of noise and SI, resulting in
\begin{align}  \label{eu:tildey_bs}
\widetilde{\mathbf y}_{b,s}= \frac{1}{\sqrt L} \int_{0}^{LT} \mathbf y_{b,s}(t)q^*(t-2\hat{\tau}_b)\,dt,
\end{align}
where ${\tau}_b \triangleq \hat{\tau}_b +\Delta\tau_b$ in which $\Delta\tau_b$ is the delay estimation error, and the sensing waveform is periodic with period $LT$ and, over one period, is given by
\begin{equation}
q(t)=\sum_{m=0}^{L-1}a_m\phi(t-mT), \qquad 0\leq t<LT,
\label{eq:q}
\end{equation}
with $q(t+rLT)=q(t)$ for any $r\in\mathbb{Z}$. Moreover,
\begin{align}
\sum_{m=0}^{L-1}|a_m|^2=L, \qquad \int_{0}^{LT}|q(t)|^2\,dt=L.
\end{align}

The following lemma approximates $\widetilde{\mathbf y}_{b,s}$ analytically.

\begin{lemma}\label{lemma1}
Under the periodic waveform model in \eqref{eq:q}, assume that $\phi(t)$ is a Nyquist pulse with a symmetric spectrum and that the delay estimation error $\Delta\tau_b$ is sufficiently small. Then, the MF output in \eqref{eu:tildey_bs} can be written as
\begin{equation}
\widetilde{\mathbf y}_{b,s} =\widetilde{\mathbf y}_{s}
+ \widetilde{\mathbf y}_{\rm SI} + \widetilde{\mathbf n}_{b}, \label{eq:mf_decomposition}
\end{equation}
such that
\begin{equation}
\widetilde{\mathbf y}_{s} \approx \sqrt{p_sL}\,\beta_0
\mathbf A(\theta_b)\mathbf a(\hat{\theta}_b) \left(1+2R''_\phi(0)\Delta\tau_b^2\right), \label{eq:mf_sensing}
\end{equation}
\begin{equation}
\widetilde{\mathbf y}_{\rm SI} \approx \sqrt{\frac{p_s}{L}}\,
\mathbf H_{\rm SI}\mathbf a(\hat{\theta}_b) c\!\left(\ell'(\hat{\tau}_b)\right).
\label{eq:mf_si}
\end{equation}
where, $R''_\phi(0)=-(2\pi\beta_{\rm rms})^2$ is the second derivative of the pulse autocorrelation function $R_\phi(\tau)$ at lag zero, and $\beta_{\rm rms}$ is the RMS bandwidth of $\phi(t)$. The pulse offset $\ell'(\hat{\tau}_b)\triangleq\operatorname{round}(2\hat{\tau}_b/T)$ corresponds to the estimated round-trip delay, and $c(n)\triangleq \sum_{\ell=0}^{L-1} a_\ell a^*_{(\ell+n)\bmod L}$is the periodic autocorrelation of the modulation sequence. The filtered noise satisfies $\widetilde{\mathbf n}_{b}\sim \mathcal{CN}(\mathbf 0,\sigma_b^2\mathbf I)$.
\end{lemma}
\begin{proof}
    See Appendix \ref{AppendixA}.
\end{proof}

The post-MF SINR is given by
\begin{align}\label{gammas}
    \gamma_s = \frac{\mathbb{E}\{ \vert \vbf_s^H \widetilde{\ybf}_\text{s} \vert^2\}}{\mathbb{E}\{ \vert \vbf_s^H \widetilde\ybf_{\rm SI} \vert^2 \} + \sigma_{b}^2 },
    \end{align}
where $\vbf_s\in\mathbb{C}^{N_r\times 1}$ is the sensing combiner.

Lemma \ref{lemma2} provides a closed-form approximation of the sensing SINR.

\begin{lemma}\label{lemma2}
    The received sensing SINR at the BS can be approximated as
\begin{align} \label{SINR_sensing_final}
    \boxed{\gamma_s \approx \frac {p_s \vert \beta_0 \vert^2 L  (\vbf_s^H \Qbf_s\vbf_s)}{\frac{p_s}{L} |c(\ell'(\hat\tau_b))|^2(\vbf_s^H \Qbf_{\rm SI}\vbf_s)+\sigma^2_b},}
\end{align}
where
\begin{align}
\Qbf_s &\triangleq \Abf(\theta_b)\abf(\theta_b)\abf(\theta_b)^H\Abf(\theta_b)^H \nonumber\\
&\quad + \sigma_{\tau_b}^2 4 R_\phi''(0)\, \Abf(\theta_b)\abf(\theta_b)\abf(\theta_b)^H\Abf(\theta_b)^H \nonumber\\
&\quad + \sigma_{\theta_b}^2\, \Abf(\theta_b)\abf'_{\theta_b} \abf_{\theta_b}'^{H}\Abf(\theta_b)^H \nonumber\\
&\quad + \frac{\sigma_{\theta_b}^2}{2}\, \Abf(\theta_b) \left(\abf''_{\theta_b}\abf(\theta_b)^H
+\abf(\theta_b)\abf_{\theta_b}''^{H}\right) \Abf(\theta_b)^H .
\label{eu_expected_c}
\end{align}
where $\sigma_{\tau_b}^2$ and $\sigma_{\theta_b}^2$ denote the modeled delay- and angle-estimation variances\footnote{In the proposed design, the corresponding diagonal entries of the angle--delay CRB submatrix are used as proxies for these variances. This approximation is most accurate for a locally unbiased and efficient estimator under accurate modeling and sufficiently high sensing SNR.}, respectively.
\end{lemma}
\begin{proof}
    See Appendix \ref{AppendixB}.
\end{proof}

\subsection{Eavesdropping Performance under Location Uncertainty}
The location uncertainty of Eve introduces stochastic perturbations in the channel estimates, affecting the wiretapping performance. Accordingly, we model the actual eavesdropping channels, taking estimation uncertainties into account. To that end, we approximate the true eavesdropping channel vectors using first-order Taylor approximations around the nominal (estimated) values of time delay and \ac{AoA} parameters. As such, the actual channel between the ISAC BS and Eve can be approximated as
\begin{align} \label{gb_actual}
  \widetilde{\gbf}_b(\tau_b, \theta_b)\approx \gbf_b(\hat\tau_b, \hat \theta_b)  + \frac{\partial\gbf_b}{\partial\tau_b}{\bigg\rvert_{\tau_b=\hat\tau_b}}\Delta\tau_b + \frac{\partial\gbf_b}{\partial\theta_b}{\bigg\rvert_{\theta_b=\hat\theta_b}}\Delta\theta_b ,
\end{align}
where $\widetilde\gbf_b(\tau_b, \theta_b)$ denotes the error-free (actual) channel between the BS and Eve, given by
\begin{align} \label{eu_truechannel_b}
 \widetilde\gbf_b(\tau_b, \theta_b) = \alpha_1(\tau_b) \abf(\theta_b),
\end{align}
where $\tau_b \triangleq \hat \tau_b + \Delta  \tau_b$ and $\theta_b \triangleq \hat \theta_b + \Delta \theta_b$ such that $\Delta\tau_b \sim \mathcal{N}(0, \sigma^2_{\tau_b})$ and $\Delta {\theta_b}\sim \mathcal{N}(0, \sigma^2_{\theta_b})$.

Similarly, the actual channel between UE$_k$ and Eve is approximated by
\begin{align} \nonumber
  \widetilde{\mathbf{g}}_{u_k}(\tau_{u_k}, \theta_{u_k})&\approx \gbf_{u_k}(\hat\tau_{u_k}, \hat \theta_{u_k})  + \frac{\partial\gbf_{u_k}}{\partial\tau_{u_k}}{\bigg\rvert_{\tau_{u_k}=\hat\tau_{u_k}}}\Delta\tau_{u_k} \\ \label{gk_actual}
 & + \frac{\partial\gbf_{u_k}}{\partial\theta_{u_k}}{\bigg\rvert_{\theta_{u_k}=\hat\theta_{u_k}}}\Delta\theta_{u_k} ,
\end{align}
where $\widetilde{\gbf}_{u_k}(\tau_{u_k}, \theta_{u_k})$ denotes the error-free link between UE$_k$ and Eve, given by
\begin{align} \label{eu:truechannel_uk}
 \widetilde\gbf_{u_k}(\tau_{u_k}, \theta_{u_k}) = \alpha_1(\tau_{u_k}) \abf(\theta_{u_k}).
\end{align}
where $\tau_{u_k} \triangleq \hat{\tau}_{u_k}+\Delta\tau_{u_k}$ and $\theta_{u_k} \triangleq \hat{\theta}_{u_k}+\Delta\theta_{u_k}$ such that $\Delta\tau_{u_k} \sim \mathcal{N}(0, \sigma^2_{\tau_{u_k}})$ and $\Delta {\theta_{u_k}}\sim \mathcal{N}(0, \sigma_{\theta_{u_k}}^2)$.

The nominal UE$_k$--Eve parameters $\hat{\tau}_{u_k}$ and $\hat{\theta}_{u_k}$ are computed from the BS estimate of Eve's location and the associated UE$_k$--Eve geometry. Because all link-parameter errors originate from the same Eve-position estimate, they are generally correlated; moreover, their marginal variances differ because the BS--Eve and UE$_k$--Eve have distinct geometries. Accordingly, the BS--Eve angle--delay covariance is propagated to each UE$_k$--Eve link through a UE-specific Jacobian, as formalized in Lemma~\ref{lemma3} of the following section.

\section{Uncertainty-aware Performance Analysis} \label{sec:CRB}
To capture the impact of uncertainty in the sensing parameters on system performance, this section develops a unified framework that models the radar parameter estimation bounds and analyzes the impact on the assumed communication channels and sensing signal $\xbf$. This uncertainty is rooted in the imperfect knowledge of Eve's location, specifically in her angle of departure $\theta_b$ and delay $\tau_b$.

\subsection{Estimation Performance via CRB}
\label{subsec:crb}

In the sensing subsystem, we are interested in characterizing the CRB for estimating these parameters:
\[ \boldsymbol{\zeta} = [\theta_b, \tau_b, \beta_r, \beta_i]^\top \in \mathbb{R}^{4\times 1}, \]
where \(\beta_0 = \beta_r + j \beta_i\) denotes the complex \ac{RCS} of target Eve. To derive the CRB for estimating $\zetabf$, we first rewrite \eqref{eu_ys} for multiple snapshots as  
\begin{align} \label{eq_multi:sense}
\Ybf_\text{b,s}= \beta_0 \Abf(\theta_b) \Xbf + \Hbf_{\rm SI}\widetilde\Xbf+\Nbf_\text{b},
\end{align}
where the radar transmit waveform matrices are  $\Xbf = [\xbf_1(t_1-2\tau_b), ...,\xbf_L(t_L-2\tau_b)] \in \mathbb{C}^{N\times L}$, $\widetilde\Xbf = [\xbf_1(t_1), ...,\xbf_L(t_L)] \in \mathbb{C}^{N\times L}$ and $\Nbf_\text{b}=[\nbf_{{\text{b}}_1},...,\nbf_{{\text{b}}_L}]\in \mathbb{C}^{N \times L}$ is AWGN matrix.

We vectorize the sensing signal $\Ybf_\text{b,s}$ as
\begin{align} \label{eq_vectorizing} \nonumber
    \ybf_\text{b,s} &= \vect(\Ybf_\text{b,s}) \\
     &= \beta_0(\Ibf_L \otimes \Abf(\theta_b))\chibf+(\Ibf_L \otimes \Hbf_{\rm SI})\widetilde\chibf+{\nbf}_\text{b} ,
\end{align}
where ${\chibf}=\vect(\Xbf)\in\mathbb{C}^{NL \times 1}$, ${\widetilde\chibf}=\vect(\widetilde\Xbf)\in\mathbb{C}^{NL \times 1}$, and $\nbf_\text{b} = \vect(\Nbf_\text{b})\sim \mathcal{CN}(0,\sigma^2_\text{b} \Ibf_{NL})$ .

To proceed, we derive the \ac{FIM} \(\mathbf{F} \in \mathbb{R}^{4 \times 4}\) associated with \(\boldsymbol{\zeta}\). After vectorization, the sensing model contains both noise and a residual SI component. We therefore approximate the effective noise covariance as 
\begin{equation}
    \mathbf{C}_b \approx \sigma^2_b\mathbf{I}_{NL}+\mathbf{C}_{\rm SI},
\end{equation}
where $\sigma^2_b\mathbf{I}_{NL}$ and $\mathbf{C}_{\rm SI}$ denote the
covariance matrices of the noise and residual SI term, respectively.

The residual SI contribution can be expressed as
\[\wbf_{\rm SI} = (\Ibf_L \otimes \Hbf_{\rm SI})\widetilde\chibf. \]
The covariance of $\wbf_{\rm SI}$ is given by
\begin{align} \nonumber
    \Cbf_{\rm SI}
    &= \mathbb{E}[ \wbf_{\rm SI} \wbf_{\rm SI}^H ] \\ \label{eu:C_SI}
    &= (\Ibf_L \otimes \Hbf_{\rm SI}) \Rbf_{\widetilde\chi} (\Ibf_L \otimes \Hbf_{\rm SI})^H,
\end{align}
where $\Rbf_{\widetilde\chi} = \mathbb{E}[\widetilde\chibf \widetilde\chibf^H]$. Let $\qbf \in \mathbb{C}^{L \times 1}$ collect the residual SI samples across the $L$ snapshots and define its temporal covariance as 
\begin{equation}\label{eq:Rq_def} \Rbf_q \triangleq \mathbb{E}[\qbf\qbf^H], \qquad [\Rbf_q]_{n,k} = \mathbb{E}\{q(t_n)\,q(t_k)^*\}. \end{equation} 
We have $[\Rbf_q]_{n,k}=c(t_n-t_k)$, and for uniform sampling $t_n=nT$, the matrix $\Rbf_q$ is Toeplitz with entries $c((n-k)T)$. Thus
\begin{equation}\label{eu:Rtildechi}
\Rbf_{\widetilde\chi} = p_s\Big(\Rbf_q \otimes \abf(\hat\theta_b)\abf(\hat\theta_b)^H\Big).
\end{equation}
Substituting \eqref{eu:Rtildechi} into \eqref{eu:C_SI}, the effective covariance can be expressed as 
\begin{equation}\label{eu:CSI_final}
\Cbf_{\rm SI} 
= p_s\Big(\Rbf_q \otimes \Hbf_{\rm SI}\abf(\hat\theta_b)\abf(\hat\theta_b)^H\Hbf_{\rm SI}^H\Big).
\end{equation}

Under this noise model, each element of \(\mathbf{F}\) is given by \cite{Kay1993}
\begin{align} \label{eu_16}
   [\mathbf{F}]_{ij} = \operatorname{tr} \left( \mathbf{C}_b^{-1} \frac{\partial \mathbf{C}_b}{\partial \zeta_i} \mathbf{C}_b^{-1} \frac{\partial \mathbf{C}_b}{\partial \zeta_j} \right) + 2 \Re \left\{ \frac{\partial \widetilde{\mubf}^H }{\partial \zeta_i} \mathbf{C}_b^{-1} \frac{\partial \widetilde{\mubf} }{\partial \zeta_j} \right\},
\end{align}
where \(\widetilde{\mubf} = \mathbb{E}[\ybf_\text{b,s}]\), and $\zeta_i$ is the $i$th element of $\pmb{\zeta}$. Conditioning on the sensing power and fixing the nominal transmit steering direction, $\mathbf C_b$ is independent of the local parameter vector $\boldsymbol{\zeta}$. Thus, the first term in \eqref{eu_16} vanishes, resulting in
\begin{align} \label{eq:FIM}
    [\mathbf{F}]_{ij} = 2\Re \left\{ \frac{\partial \widetilde{\mubf}^H }{\partial \zeta_i} \mathbf{C}_b^{-1}
    \frac{\partial\widetilde{\mubf}}{\partial \zeta_j} \right\}, \quad i,j \in \{1,2,3,4\},
\end{align}

The derivatives of \(\widetilde{\mubf}\) with respect to the parameters are
\begin{subequations} \label{eq:mu-derivatives}
\begin{align}
\mathbf d_{\theta_b} \triangleq\frac{\partial \widetilde{\mubf}}{\partial \theta_b} &= {\beta}_0 \dot\Abf_L \chibf, \\
\mathbf d_{\tau_b} \triangleq\frac{\partial \widetilde{\mubf}}{\partial \tau_b}
&=\beta_0\left(\dot{\mathbf A}_{L,\tau}\boldsymbol{\chi}
+\mathbf A_L\dot{\boldsymbol{\chi}}\right),\\
\mathbf d_{\beta_r} \triangleq\frac{\partial \widetilde{\mubf}}{\partial \beta_r} &= \Abf_L(\theta_b)\chibf, \\
\mathbf d_{\beta_i} \triangleq\frac{\partial \widetilde{\mubf}}{\partial \beta_i} &= j \Abf_L(\theta_b)\chibf,
\end{align}
\end{subequations}
where $\dot\Abf_L=\Ibf_L\otimes\dot{\mathbf{A}}$, $\Abf_L=\Ibf_L\otimes \Abf(\theta_b)$, and $\dot{\mathbf A}_{L,\tau}\triangleq\frac{\partial\mathbf A_L}{\partial\tau_b}$.
The other derivatives are defined as
$\dot{\mathbf{A}}=\frac{\partial\Abf(\theta_b)}{\partial\theta_b}$ and $\dot{\chibf} = \frac{\partial\vect(\Xbf)}{\partial\tau_b}$, with
\begin{subequations}
\begin{align}
\left.\frac{\partial \Abf(\theta)}{\partial \theta}\right|_{\theta=\theta_b}
&= \alpha_2(\tau_b)\left(\dot{\abf}(\theta_b)\abf(\theta_b)^H
+ \abf(\theta_b)\dot{\abf}(\theta_b)^H\right),  \\ \nonumber
\frac{\partial\vect(\Xbf)}{\partial \tau_b}
&=-2\left.\frac{\partial\vect(\Xbf)}{\partial t}
\right|_{t_\ell=t_\ell-2\tau_b}\\
&=-2\,\vect\!\left(\left[\dot{\xbf}_1(t_1-2\tau_b),\dots,\dot{\xbf}_L(t_L-2\tau_b)\right]\right).
\end{align}
\end{subequations}

Let $\mathbf D\triangleq 
\left[\mathbf d_{\theta_b},\,
\mathbf d_{\tau_b},\,
\mathbf d_{\beta_r},\,
\mathbf d_{\beta_i}
\right].$ The FIM in \eqref{eq:FIM} can be then written compactly as
\begin{equation}\label{eq:FIM_compact}
\boxed{\mathbf F=2\Re\!\left\{\mathbf D^{H}\mathbf C_b^{-1}\mathbf D \right\}.}
\end{equation}
Equivalently, its $(i,j)$th entry is
\begin{equation}
[\mathbf F]_{i,j}=2\Re\!\left\{\mathbf d_i^{H}\mathbf C_b^{-1}\mathbf d_j\right\}, \qquad i,j\in\{1,2,3,4\}.
\end{equation}

The FIM in \eqref{eq:FIM_compact} provides the classical deterministic-parameter CRB, conditioned on the adopted waveform, channel model, and residual-SI covariance. To model the estimation errors in the subsequent robust design, we approximate their second-order statistics by the angle--delay submatrix of the inverse FIM:
\begin{equation}
\boldsymbol{\Sigma}_b \triangleq
\begin{bmatrix}
\sigma_{\theta_b}^2 & \sigma_{(\theta\tau)_b}\\
\sigma_{(\theta\tau)_b} & \sigma_{\tau_b}^2
\end{bmatrix}
\approx
[\mathbf{F}^{-1}]_{1:2,1:2}.
\end{equation}
where $\sigma_{(\tau\theta)_b}=\sigma_{(\theta\tau)_b}$ denotes the cross-covariance between the delay- and angle-estimation
errors. We now derive how the BS--Eve angle--delay uncertainty propagates to the UE$_k$--Eve link parameters through the geometry. Let $\varphi_b$ be the known BS--Eve azimuth angle. The Cartesian position associated with the BS--Eve elevation angle
$\theta_b$ and one-way delay $\tau_b$ is
\begin{equation}
\mathbf d_e(\theta_b,\tau_b)=\mathbf d_b+c\tau_b
\begin{bmatrix}
\cos\theta_b\cos\varphi_b\\\cos\theta_b\sin\varphi_b\\\sin\theta_b
\end{bmatrix}.
\label{eq:eve_position_mapping}
\end{equation}
For UE$_k$, define
\begin{equation}
\mathbf h_k(\theta_b,\tau_b) \triangleq
\begin{bmatrix}
\displaystyle\sin^{-1}\!\left(\frac{[\mathbf d_e(\theta_b,\tau_b)-\mathbf d_k]_z}{\|\mathbf d_e(\theta_b,\tau_b)-\mathbf d_k\|}\right)\\[2mm]\displaystyle \frac{\|\mathbf d_e(\theta_b,\tau_b)-\mathbf d_k\|}{c}
\end{bmatrix}
=\begin{bmatrix}
\theta_{u_k}\\
\tau_{u_k}
\end{bmatrix},
\label{eq:geometry_mapping}
\end{equation}
where $[\cdot]_z$ extracts the vertical component, so that the nominal UE$_k$--Eve parameters are $[\hat\theta_{u_k},\hat\tau_{u_k}]^\top=\mathbf h_k(\hat\theta_b,\hat\tau_b).$
\begin{lemma}\label{lemma3}
To first order in the estimation errors around $(\hat{\theta}_b,\hat{\tau}_b)$, the UE$_k$--Eve parameter errors satisfy
\begin{equation}
\begin{bmatrix}
\Delta\theta_{u_k}\\\Delta\tau_{u_k}
\end{bmatrix}
\simeq \mathbf J_{u_k}
\begin{bmatrix}
\Delta\theta_b\\\Delta\tau_b
\end{bmatrix},
\qquad\mathbf J_{u_k}
\triangleq\left.\frac{\partial\mathbf h_k(\theta_b,\tau_b)}{\partial[\theta_b,\tau_b]} \right|_{(\hat{\theta}_b,\hat{\tau}_b)},
\label{eq:UE_error_mapping}
\end{equation}
and, if $[\Delta\theta_b,\Delta\tau_b]^\top$ is zero-mean with covariance $\boldsymbol\Sigma_b$, the UE$_k$--Eve angle--delay covariance is
\begin{equation}
\boldsymbol{\Sigma}_{u_k}\simeq \mathbf J_{u_k}\boldsymbol{\Sigma}_b\mathbf J_{u_k}^{T}
=
\begin{bmatrix}
\sigma_{\theta_{u_k}}^2 & \sigma_{(\theta\tau)_{u_k}}\\\sigma_{(\tau\theta)_{u_k}} &\sigma_{\tau_{u_k}}^2
\end{bmatrix}.
\label{eq:UE_uncertainty_covariance}
\end{equation}
\end{lemma}
\begin{proof}
Equation \eqref{eq:UE_error_mapping} is the first-order Taylor expansion of \eqref{eq:geometry_mapping} around $(\hat\theta_b,\hat\tau_b)$, and \eqref{eq:UE_uncertainty_covariance} follows by taking the covariance of the linearized errors.
\end{proof}

\subsection{Received SINR at Eve}

Under the eavesdropper-knowledge assumption stated earlier, the BS sensing waveform is treated as interference when evaluating Eve's SINR. In addition, we consider the following cases from the SIC perspective: 
\begin{itemize}
    \item Case 1 (No SIC): When the received power of the UE signals is comparable, or when neither signal can be reliably decoded and canceled, Eve decodes each message while treating the other as interference. Thus, the SINR for decoding the $k$th UE at Eve is given by
    \begin{align} \label{eu_case1_Eve}
      \gamma_{e, k} = \frac{  \mathbb{E}\{\vert \widetilde{\gbf}_{u_k}^H\widetilde \wbf_k \vert^2 \}}{\sum_{m \neq k} 
      \mathbb{E}\{\vert \widetilde{\gbf}_{u_m}^H \widetilde \wbf_m \vert^2 \} + \mathbb{E}\{\vert \widetilde{\gbf}_{b}^H \xbf \vert^2\}+ \sigma^2_e}.
    \end{align}
     \item Case 2 (With SIC): If one signal is significantly stronger than the other, Eve applies SIC to decode the dominant signal first, denoted by UE$_k$. Then, the SINR for decoding the weaker signal UE$_m$ after canceling the dominant signal is given by
     \begin{align} \label{eu_case2_Eve}
      \gamma_{e, m} = \frac{ \mathbb{E}\{\vert \widetilde{\gbf}_{u_m}^H  \widetilde \wbf_m \vert^2\} }{\mathbb{E}\{\vert \widetilde{\gbf}_{b}^H \xbf \vert^2\}+ \sigma^2_e}, ~m,k\in\{1,2\},\; m\neq k.
    \end{align} 
\end{itemize}
In the Eve decoding analysis, we consider Case 2 and assume that the received signal power of UE$_2$ at Eve is stronger than that of UE$_1$, such that \( \vert \widetilde{\gbf}_{u_2}^H \widetilde{\wbf}_2 \vert^2 \geq \lambda_{\text{\tiny SIC}} \vert \widetilde{\gbf}_{u_1}^H \widetilde{\wbf}_1 \vert^2 \), where \( \lambda_{\text{\tiny SIC}} \geq 1 \) is a predefined power-ordering factor for SIC at Eve. Under this condition, Eve attempts to decode the dominant signal from UE$_2$ first and, after its successful cancellation, attempts to recover UE$_1$'s message via SIC. To limit the information leakage at both decoding stages, we consider the maximum effective SINR across Eve's SIC decoding stages, given by
    \begin{align} \nonumber
       \gamma_e &= \max\Bigg\{\frac{ \mathbb{E} \{ \vert \widetilde{\gbf}_{u_2}^H  \widetilde \wbf_2 \vert^2 \} }{ \mathbb{E}\{\vert \widetilde{\gbf}_{u_1}^H \widetilde \wbf_1 \vert^2 \} + \mathbb{E}\{\vert \widetilde{\gbf}_{b}^H  \xbf \vert^2\}+ \sigma^2_e}, \\ \label{eu_case2_Eve_max}
      & \qquad \qquad \qquad \qquad \frac{ \mathbb{E} \{ \vert \widetilde{\gbf}_{u_1}^H  \widetilde \wbf_1 \vert^2\}}{\mathbb{E}\{\vert \widetilde{\gbf}_{b}^H \xbf \vert^2\}+ \sigma^2_e}\Bigg\} .
    \end{align}
A closed-form analytical expression for $\gamma_e$ is presented in the following lemma.
\begin{lemma}\label{lemma4}
Accounting for the uncertainty in her location, the received SINR at Eve is approximately   
\begin{align} \label{eu_case2_Eve_xNew}
    \boxed{\gamma_e \approx \max\left\{\frac{ \widetilde\wbf_2^H \Gbf_{u_2} \widetilde{\wbf}_2}
  { \widetilde\wbf_1^H \Gbf_{u_1} \widetilde\wbf_1+  p_s q_b+ \sigma^2_e}, \frac{\widetilde\wbf_1^H \Gbf_{u_1} \widetilde{\wbf}_1}
   { p_s q_b + \sigma^2_e}\right\},}
\end{align}
where \begin{align} \nonumber
     q_b &= \big( \alpha_1(\hat\tau_b)^2 \|\abf(\hat\theta_b)\|^4 + \sigma^2_{\tau_b} \alpha_{\tau_b}'^2 \|\abf(\hat\theta_b)\|^4 \\ \nonumber
    & \quad + \sigma_{(\tau\theta)_b} \alpha_{\tau_b}' \alpha_1(\hat\tau_b)^H \abf_{\theta_b}'^H \abf(\hat\theta_b)   \|\abf(\hat\theta_b)\|^2 \\ \nonumber
    & \quad + \sigma_{(\theta\tau)_b} \alpha_{\tau_b}'^H \alpha_1(\hat\tau_b) \abf(\hat\theta_b)^H \abf_{\theta_b}' \|\abf(\hat\theta_b)\|^2  \\ \label{eu_qb}
    & \quad + \sigma^2_{\theta_b} \alpha_1(\hat\tau_b)^2 \vert \abf(\hat\theta_b)^H \abf'_{\theta_b} \vert^2 \big).
\end{align} 
\begin{align} \nonumber
    \Gbf_{u_k}&\triangleq
    \mathbb{E}\!\left\{\widetilde{\gbf}_{u_k}\widetilde{\gbf}_{u_k}^{H}\right\}\\ \nonumber
    &\approx 
    \alpha_1(\hat\tau_{u_k})^2
    \abf(\hat\theta_{u_k})\abf(\hat\theta_{u_k})^H
    +\sigma^2_{\theta_{u_k}}
    \alpha_1(\hat\tau_{u_k})^2
    \abf'_{\theta_{u_k}}\abf'^{H}_{\theta_{u_k}} \\ \nonumber
    &\quad
    +\sigma_{(\tau\theta)_{u_k}}
    \alpha_1(\hat\tau_{u_k})\alpha'_{\tau_u}
    \abf(\hat\theta_{u_k})\abf'^H_{\theta_{u_k}}\\ \nonumber
    &\quad
    +\sigma_{(\theta\tau)_{u_k}}
    \alpha'_{\tau_u}\alpha_1(\hat\tau_{u_k})
    \abf'_{\theta_{u_k}}\abf(\hat\theta_{u_k})^H\\
    &\quad+\sigma^2_{\tau_{u_k}}\alpha'^2_{\tau_u}\abf(\hat\theta_{u_k})\abf(\hat\theta_{u_k})^H.
    \label{eq:Guk}
\end{align}
\end{lemma}
\begin{proof}
    See Appendix \ref{AppendixD}.
\end{proof}

\allowdisplaybreaks[4]
\section{Joint Optimization Framework for Secure NOMA--ISAC} \label{sec:Joint_Opt}
In this section, we design a secure uplink transmission strategy for the considered NOMA--ISAC system by jointly optimizing beamforming and power allocation. The objective is to maximize the achievable communication rate and sensing performance while securing the uplink transmissions against eavesdropping. To this end, we first formulate the resource allocation as a nonconvex optimization problem to improve communication and sensing performance, given a transmit power budget and information leakage constraints. We then propose an efficient iterative algorithm to tackle it.

\subsection{Problem Formulation}  \label{sec:probelm_formulation}
We consider the following utility function
\begin{align}\label{utility}
   U(\gamma_k, \gamma_s; a_k, a_s)= \sum_{k=1}^2 a_k\log_2(1\!+\!\gamma_k) +a_s\log_2(1\!+\!\gamma_s)
\end{align}
where the coefficients $a_1$ and $a_2$ are scaling factors that determine the importance of the individual UE data rates (i.e., the uplink communication performance), and $a_s\geq0$ weights the importance of the sensing criterion \footnote{Notice that the term $\log_2(1+\gamma_s)$ is adopted as a dimensionless, monotonically increasing sensing utility. The logarithm compresses the dynamic range of the sensing SINR and facilitates its weighted combination with the communication-rate terms.}. If $a_s \gg a_k$, then the sensing performance becomes our priority; otherwise, when $a_s=0$, only communication performance is weighted in the objective, while the security constraints remain enforced.

We formulate the joint \ac{OP} as
\begin{subequations}
\begin{align} 
\mathcal{P}_0:~\max_{\{\widetilde\wbf_k,\vbf_k,\vbf_s,p_s\}} & \sum_{k=1}^2 a_k\log_2(1\!+\!\gamma_k) + a_s\log_2(1\!+\!\gamma_s) \nonumber\\
\text{s.t.}\quad
\mathbf{C1}:&\quad \| \widetilde \wbf_k \|^2 \leq P^{\mathrm{max}}_u, \quad k=1,2   \label{C1} \\
\mathbf{C2}:&\quad   p_s \geq 0, \label{C2} \\
\mathbf{C3}:&\quad \mathbb{E}\{ \| \xbf \| ^2 \} \leq P^\mathrm{max}_s,\label{C3}\\
\mathbf{C4}:&\quad \gamma_e \leq \gamma_{\mathrm{th}},  \label{C4} \\  \label{C5}
\mathbf{C5}: &\quad 
    \widetilde\wbf_2^H \Gbf_{u_2} \widetilde\wbf_2 \ge \lambda_{\text{\tiny SIC}} \, \widetilde\wbf_1^H \Gbf_{u_1} \widetilde\wbf_1,
\end{align}
\end{subequations}
where {\bf C1} ensures that the UEs' transmit power does not exceed the maximum power budget $P^{\mathrm{max}}_u$, {\bf C2} guarantees non-negative sensing power, {\bf C3}  imposes an upper bound $P^\mathrm{max}_s$ on the total radar signal transmit power, {\bf C4} limits information leakage using secrecy threshold $\gamma_{\mathrm{th}}$, and {\bf C5} enforces the received-power ordering required for SIC feasibility at Eve.

\begin{remark}
The SINR expressions in Lemma~\ref{lemma4} are based on second-order average received powers under the adopted channel-uncertainty model. Hence, constraint {\bf C4}, i.e., $\gamma_e\leq\gamma_{\rm th}$, limits an uncertainty-aware statistical SINR surrogate rather than guaranteeing that Eve's instantaneous SINR remains below $\gamma_{\rm th}$ for every channel-error realization. The parameter $\gamma_{\rm th}$ therefore serves as a prescribed leakage-SINR cap, which bounds the modeled leakage of user $k$ by $\log_2(1+\gamma_{\rm th})$ bits/s/Hz. 
\end{remark}

Problem $\mathcal{P}_0$ is nonconvex because of the coupled fractional SINR terms in the objective function, and the nonconvex SIC-ordering and leakage constraints with respect to (w.r.t.) the optimization variables $\widetilde\wbf_k,\vbf_k,\vbf_s, \text{ and } p_s$. In the following, we present an algorithm to efficiently tackle the \ac{OP}.

\subsection{Proposed Solution} \label{sec:Proposed Solution}
To tackle the nonconvex \ac{OP}, we apply an \ac{AO}, where we split $\mathcal{P}_0$ into two subproblems: 
\begin{itemize}
    \item \textit{Subproblem 1}: Optimize $\vbf_k$ and $\vbf_s$ while keeping the other variables fixed.
    \item \textit{Subproblem 2}: Jointly optimize the beamformer vectors $\widetilde\wbf_k,~\forall k$, and the sensing power $p_s$.  
\end{itemize}
These two subproblems are alternately solved until convergence, starting from an initial feasible point.

\subsubsection{Solution to Subproblem 1}
In problem $\mathcal{P}_0$, the receive combiners $\vbf_k$ and $\vbf_s$ appear exclusively in the objective function via the SINR terms in the form of a sum of logarithmic expressions.  Since the logarithms in the objective are strictly increasing, maximizing the objective w.r.t. these variables is equivalent to directly maximizing their respective SINRs. The SINRs \eqref{eu_SINR_user1_xN}, \eqref{eu_SINR_user2_xN}, and \eqref{SINR_sensing_final}  can be rewritten as Rayleigh quotients, i.e., ratios of quadratic forms:
\begin{subequations}\label{eq:sinr_rayleigh}
\begin{align}
\gamma_1 &=\frac{\vbf_1^H \widetilde{\Abf}\vbf_1}{\vbf_1^H \widetilde{\Bbf}\vbf_1}, \\*
\gamma_2 &=\frac{\vbf_2^H \widetilde{\Cbf}\vbf_2}{\vbf_2^H \widetilde{\Dbf}\vbf_2}, \\*
\gamma_s &=\frac{\vbf_s^H \widetilde{\Ebf}\vbf_s}{\vbf_s^H \widetilde{\Fbf}\vbf_s},
\end{align}
\end{subequations}
where 
\begin{align}
\widetilde{\Abf} &=\Hbf_{u_1}\widetilde{\wbf}_1\widetilde{\wbf}_1^H\Hbf_{u_1}^H, \notag\\ \widetilde{\Bbf} &=\Hbf_{u_2}\widetilde{\wbf}_2\widetilde{\wbf}_2^H\Hbf_{u_2}^H + p_s\Qbf_{\rm SI}+ \sigma_b^2\Ibf, \notag\\
\widetilde{\Cbf} &=\Hbf_{u_2}\widetilde{\wbf}_2\widetilde{\wbf}_2^H\Hbf_{u_2}^H,\quad
\widetilde{\Dbf} = p_s\Qbf_{\rm SI}+ \sigma_b^2\Ibf, \notag\\
\widetilde{\Ebf} &= p_s|\beta_0|^2L\Qbf_s,\quad
\widetilde{\Fbf} = \frac{p_s}{L}|c(\ell'(\hat\tau_b))|^2\Qbf_{\rm SI}+ \sigma_b^2\Ibf.
\end{align}

Since all of the above matrices are Hermitian and the matrices in the denominator are positive definite, the optimal SINRs can be expressed as
\begin{subequations}\label{eu_precoders}
\begin{align}
\gamma_1 &=
\lambda_{\max}\bigl(\widetilde{\Bbf}^{-1}\widetilde{\Abf}\bigr),
\quad \|\vbf_1\|^2 = 1, \\
\gamma_2 &=
\lambda_{\max}\bigl(\widetilde{\Dbf}^{-1}\widetilde{\Cbf}\bigr),
\quad \|\vbf_2\|^2 = 1, \\
\gamma_s &=
\lambda_{\max}\bigl(\widetilde{\Fbf}^{-1}\widetilde{\Ebf}\bigr),
\quad \|\vbf_s\|^2 = 1.
\end{align}
\end{subequations}
where $\lambda_{\max}(\cdot)$ denotes the largest eigenvalue of its matrix argument. Since the receive combiners are not involved in any constraints, they can be updated in closed form at each iteration, given fixed values of $\wbf_k$ and $p_s$. Therefore, the optimal receive combiners for this subproblem are the normalized generalized eigenvectors corresponding to the largest generalized eigenvalue of each matrix:
\begin{subequations}\label{eu_opt:precoder}
\begin{align}
\vbf_1^* &= \operatorname{eigvec}_{\max}(\widetilde{\Abf},\widetilde{\Bbf}),\\
\vbf_2^* &= \operatorname{eigvec}_{\max}(\widetilde{\Cbf},\widetilde{\Dbf}),\\
\vbf_s^* &= \operatorname{eigvec}_{\max}(\widetilde{\Ebf},\widetilde{\Fbf}).
\end{align}
\end{subequations}

\subsubsection{Solution to Subproblem 2: Joint Optimization of \texorpdfstring{$\widetilde{\wbf}_k, p_s$}{wk, ps}}
Given fixed combiners $\vbf_k$, $\vbf_s$, we jointly optimize the beamforming vectors $\widetilde\wbf_k$ and sensing power $p_s$. In light of this, we reformulate problem $\mathcal{P}_0$ as

\begin{subequations}\label{PF}
% ---------- Objective ----------
\begin{align}
\mathcal{P}:~\max_{\widetilde{\mathbf{W}}_k,\,p_s}\quad &
a_1 \log_2\left(1+\frac{\operatorname{tr}(\widetilde{\mathbf{W}}_1\mathbf{M}_1)}
{\operatorname{tr}(\widetilde{\mathbf{W}}_2\mathbf{N}_2)
+p_s q_1+\sigma_b^2}\right) \nonumber\\
&+a_2 \log_2\left(1+\frac{\operatorname{tr}(\widetilde{\mathbf{W}}_2\mathbf{M}_2)}
{p_s q_2+\sigma_b^2}\right)\nonumber\\
&+a_s \log_2\left(1+\frac{p_s q_s}{p_s q_{\rm SI}+\sigma_b^2}\right)\label{eu_PF}
\end{align}
\vspace{-2mm}
% ---------- Constraints ----------
\begin{flalign}
\text{s.t.}\quad
&\tr(\widetilde{\Wbf}_k) \leq P_U^{\max},\quad \forall k \label{cst1}\\
&\tr(\Gbf_{u_2}\widetilde{\Wbf}_2) \leq \gamma_{\rm th}
\bigl(\tr(\Gbf_{u_1}\widetilde{\Wbf}_1) +p_s q_b+\sigma_e^2\bigr), \label{cst2}\\
&\tr(\Gbf_{u_1}\widetilde{\Wbf}_1)\leq \gamma_{\rm th}\bigl(p_s q_b+\sigma_e^2\bigr),\\
&\tr(\Gbf_{u_2}\widetilde{\Wbf}_2)-\lambda_{\text{\tiny SIC}}\tr(\Gbf_{u_1}\widetilde{\Wbf}_1)\geq 0,\\
& 0\leq p_s\leq P_s^{\max},\label{cst3}\\
&\widetilde{\Wbf}_k\succeq 0 , \label{cst4}\\
&\operatorname{rank}(\widetilde{\Wbf}_k)=1,\quad\forall k \label{cst5}
\end{flalign}
\end{subequations}
where we define $\widetilde\Wbf_k= \widetilde\wbf_k\widetilde\wbf_k^H$ and thus the PSD and rank-one constraints \eqref{cst4}--\eqref{cst5} are introduced. $\Gbf_{u_k}$, $k\in\{1,2\}$, denotes the uncertainty-aware channel covariance matrix derived in \eqref{eq:Guk}, and the remaining coefficients are defined as follows:
\[
\begin{aligned}
    \Mbf_1 &=  \Hbf_{u_1}^H \vbf_1 \vbf_1^H \Hbf_{u_1},\quad
    \Mbf_2 =  \Hbf_{u_2}^H \vbf_2 \vbf_2^H \Hbf_{u_2},\\
    \Nbf_2 &=  \Hbf_{u_2}^H \vbf_1 \vbf_1^H \Hbf_{u_2},\\
    q_k &=   \vbf_k^H \Qbf_{\rm SI}\vbf_k,\\
    q_{\rm SI}&=\frac{|c(\ell'(\hat\tau_b))|^2}{L}\vbf_s^H\Qbf_{\rm SI}\vbf_s,
    \quad
    q_s =|\beta_0|^2L\,\vbf_s^H\Qbf_s\vbf_s.
\end{aligned}
\]

Problem $\mathcal{P}$ is challenging to solve due to the nonconvex objective function.  To address this issue, we employ the quadratic transform technique introduced in \cite{FP2018Shen}, which enables the transformation of each fractional term into a concave surrogate function through the introduction of auxiliary scalar variables. Specifically, let $y_1$, $y_2$, and $y_s$ denote auxiliary variables associated with each SINR term in the objective. Given the monotonicity and concavity of the log function $\log(1+x),~x\geq0$, the original problem can be approximately relaxed to the following quadratic-transform problem, which is solved iteratively in a \ac{SCA} framework \cite{Boyd2006}:
\begin{subequations} \label{eq:sca_relaxed}
\begin{align}
\mathcal{P}_1:~\max_{\widetilde{\mathbf{W}}_k,\, p_s} \quad & 
a_1 \log_2\Bigg(1 + 2 y_1^{(t)} \sqrt{\operatorname{tr}(\widetilde{\mathbf{W}}_1 \mathbf{M}_1)} \nonumber \\
& \quad - \left(y_1^{(t)}\right)^2 \left( \operatorname{tr}(\widetilde{\mathbf{W}}_2 \mathbf{N}_2) + p_s q_1 + \sigma_b^2 \right) \Bigg) \nonumber \\
& \hspace{-20mm}+ a_2 \log_2\Bigg(1 + 2 y_2^{(t)} \sqrt{\operatorname{tr}(\widetilde{\mathbf{W}}_2 \mathbf{M}_2)} - \left(y_2^{(t)}\right)^2 \left( p_s q_2 + \sigma_b^2 \right) \Bigg) \nonumber \\
&\hspace{-20mm} + a_s \log_2\Bigg(1 + 2 y_s^{(t)} \sqrt{p_s q_s} - \left(y_s^{(t)}\right)^2  \left(  p_s q_{\rm SI} + \sigma_b^2 \right) \Bigg)\hspace{-3mm} \label{eq:sca_obj}\\
\text{s.t.} \quad 
& \eqref{cst1}\text{--}\eqref{cst4}. \label{csts_cvx}
\end{align}
\end{subequations}
In the above reformulation, we use SDR to remove the rank-one constraint \eqref{cst5} and allow for a tractable optimization \cite{SDP2010Quan}. With fixed $y_1$, $y_2$, and $y_s$,  problem $\mathcal{P}_1$ is convex. 
The auxiliary variables are updated at each iteration as follows:
\begin{subequations} \label{eq:aux_updates}
\begin{align}
y_1^{(t)} &= \frac{\sqrt{\operatorname{tr}(\widetilde{\mathbf{W}}_1^{(t)} \mathbf{M}_1)}}
{\operatorname{tr}(\widetilde{\mathbf{W}}_2^{(t)} \mathbf{N}_2) + p_s^{(t)} q_1 + \sigma_b^2}, \\
y_2^{(t)} &= \frac{\sqrt{\operatorname{tr}(\widetilde{\mathbf{W}}_2^{(t)} \mathbf{M}_2)}}{p_s^{(t)} q_2 + \sigma_b^2}, \\
y_s^{(t)} &= \frac{\sqrt{p_s^{(t)} q_s}}{p_s^{(t)} q_{\rm SI} + \sigma_b^2}.
\end{align}
\end{subequations}
For fixed combiners $\vbf_k$, $\vbf_s$ and fixed uncertainty-dependent coefficients $\Gbf_{u_k}$, $q_b$, and $q_s$, these iterations yield a nondecreasing surrogate objective and converge to a stationary point of the relaxed problem \cite{FP2018Shen}.

Problem $\mathcal{P}_1$, despite being convex, is not acceptable to the CVX solver in its current form. To tackle this issue, we use the epigraph reformulation to make the convex problem obey the Disciplined Convex Programming (DCP) ruleset. Let \( x_1, x_2, x_s \in \mathbb{R}_+ \) be non-negative slack variables, and also define \( r_1, r_2, r_s \) as additional slack variables. We rewrite problem $\mathcal{P}_1$ equivalently as
\begin{subequations} \label{eq:sca_epigraph_relax}
\begin{align}
\mathcal{P}_2:~&\max_{\widetilde{\mathbf{W}}_k, p_s,\, \{x_i\}, \{r_i\}} 
\begin{aligned}
a_1 r_1 &+ a_2 r_2+ a_s r_s\nonumber
\end{aligned} \label{eq:sca_obj_epi} \\[5pt]
\text{s.t.} \quad 
& \log_2(1 + x_1) \ge r_1,\\
& \log_2(1 + x_2) \ge r_2,\\
& \log_2(1 + x_s) \ge r_s,\\
& \hat{\gamma}_1 \ge x_1,~\hat{\gamma}_2 \ge x_2,~\hat{\gamma}_s \ge x_s,~\eqref{csts_cvx},
\end{align}
\end{subequations}
where $\hat{\gamma}_1$, $\hat{\gamma}_2$, and $\hat{\gamma}_s$ are concave functions of the optimization variables $\widetilde{\mathbf{W}}_1$, $\widetilde{\mathbf{W}}_2$, and $p_s$, defined respectively as
\begin{subequations}
\begin{align}
\hat{\gamma}_1
&= 2y_1^{(t)}
\sqrt{\tr\!\left(\widetilde{\Wbf}_1\Mbf_1\right)} 
\nonumber\\
&\quad -\left(y_1^{(t)}\right)^2 \left(\tr\!\left(\widetilde{\Wbf}_2\Nbf_2\right)
+p_s q_1+\sigma_b^2\right), \label{eq:gamma_hat_1}
\\
\hat{\gamma}_2 &=2y_2^{(t)}\sqrt{\tr\!\left(\widetilde{\Wbf}_2\Mbf_2\right)} -\left(y_2^{(t)}\right)^2 \left(p_s q_2+\sigma_b^2\right),\label{eq:gamma_hat_2}
\\
\hat{\gamma}_s &=2y_s^{(t)} \sqrt{p_s q_s} -\left(y_s^{(t)}\right)^2 \left(p_s q_{\rm SI}+\sigma_b^2\right). \label{eq:gamma_hat_s}
\end{align}
\label{eq:gamma_hat}
\end{subequations}

With fixed receive combiners and uncertainty-dependent quantities, $\mathcal{P}_2$ can be efficiently solved using standard conic-optimization tools (e.g., CVXPY). Alternating closed-form auxiliary-variable updates with the solution of $\mathcal P_2$ results in a nondecreasing objective for the SDR-relaxed quadratic-transform subproblem. The overall procedure is summarized in Algorithm~\ref{algo}. Since the uncertainty quantities are recomputed at each outer iteration (steps 4 and 5), the outer loop is a fixed-point iteration for which monotonicity is not guaranteed; in all our experiments, it converged within a few iterations.

\begin{algorithm}[tpb]
\caption{Proposed iterative algorithm for power and beamforming optimization}
\label{algo}
    \begin{algorithmic}[1]
        \REQUIRE $P^\mathrm{max}_u$, $P^\mathrm{max}_s$, $\gamma_\mathrm{th}$, $\Hbf_{u_k}$, $\sigma^2_b$.
         \ENSURE $\vbf^\ast_k, \widetilde\wbf^\ast_k,~\forall k=\{1, 2\}$, $\vbf^\ast_s$, and $p^\ast_s$.
        \STATE \textbf{Initialization} 
        \STATE Set feasible initial points $\widetilde\Wbf_k^{(0)}$, $p_s^{(0)}$, $\vbf^{(0)}_k, \vbf^{(0)}_s$, $y_k^{(0)}$, $y_s^{(0)}$,  and iteration index $t \gets 0$.
        \REPEAT
        \STATE Using the current sensing power $p_s^{(t)}$, compute the FIM $\mathbf F^{(t)}$ in \eqref{eq:FIM_compact} and form the BS--Eve angle--delay covariance proxy $\boldsymbol{\Sigma}_b^{(t)} =[(\mathbf F^{(t)})^{-1}]_{1:2,1:2}$.
        \STATE Propagate $\boldsymbol{\Sigma}_b^{(t)}$ through \eqref{eq:UE_uncertainty_covariance} to obtain $\boldsymbol{\Sigma}_{u_k}^{(t)}$, then update $\mathbf Q_s^{(t)}$, $q_b^{(t)}$, and $\mathbf G_{u_k}^{(t)}$.
        \STATE Compute optimal $\vbf^{(t+1)}_k, \vbf^{(t+1)}_s$ using \eqref{eu_opt:precoder}.
        \STATE Using $\vbf_k^{(t+1)}$ and $\vbf_s^{(t+1)}$, update the auxiliary variables $y_k^{(t+1)}$ and $y_s^{(t+1)}$ according to~\eqref{eq:aux_updates}.
        \STATE Given $\{y_k^{(t+1)}$, $y_s^{(t+1)}\}$, solve $\mathcal{P}_2$ to obtain $\widetilde \Wbf^{(t+1)}_k$ and $p^{(t+1)}_s$.
        \STATE Set $t \gets t + 1$
        \UNTIL Convergence: Fractional increase in the objective function $f$ is below threshold $\varepsilon$, i.e. $\frac{|f^{(t+1)} - f^{(t)}|}{f^{(t)}} \le \varepsilon$.
\STATE Extract a rank-one beamformer approximation solution $\widetilde\wbf^{(t+1)}_k$ through the dominant eigenmode of $\widetilde\Wbf_k^{(t+1)}$: \[\tilde\wbf_k^\ast \approx \sqrt{\lambda_{\max}(\tilde\Wbf_k^\ast)}
\,\operatorname{eigvec}_{\max}(\tilde\Wbf_k^\ast).\]
\RETURN 
 $\vbf^\ast_k \gets \vbf^{(t+1)}_k$, $\widetilde\wbf^\ast_k \gets \widetilde\wbf^{(t+1)}_k$, 
 $\vbf^\ast_s \gets \vbf^{(t+1)}_s$, and $p^\ast_s \gets p^{(t+1)}_s$.
        \end{algorithmic}
\end{algorithm}

%%%%%%%%%%%%%%%%%%%%%%%%%%%%%%%%%%%%%%%%%%%%%%%%%%%%%%%%%%%%%
%------------------------  Simulations ---------------------%
%%%%%%%%%%%%%%%%%%%%%%%%%%%%%%%%%%%%%%%%%%%%%%%%%%%%%%%%%%%%%
\allowdisplaybreaks[0]
 \section{Simulation Results} \label{sec:simulations}
\begin{table}[t]
\caption{Simulation Parameters}
\label{tab:sim_params}
\centering
\footnotesize
\setlength{\tabcolsep}{2.5pt}
\renewcommand{\arraystretch}{0.95}
\resizebox{\columnwidth}{!}{%
\begin{tabular}{@{}lr@{\hspace{4pt}\vrule\hspace{4pt}}lr@{}}
\toprule
\textbf{Parameter (Symbol)} & \textbf{Value} &
\textbf{Parameter (Symbol)} & \textbf{Value} \\
\midrule
Number of BS antennas ($N$)     & $8$ & Noise figure ($\mathrm{NF}$)  & $5$~dB \\
Number of UE antennas ($M$)     & $4$ & Noise bandwidth ($B$)         & $100$~MHz \\
Carrier frequency ($f_c$)       & $6$~GHz  & Noise power ($\sigma^2$) & $-89$~dBm \\
Sampling frequency ($f_s$)      & $10$~GHz & Max UE power ($P_U^{\max}$)   & $1$~W \\
Pulse interval ($T$)            & $10$~ns  & Max sensing power ($P_s^{\max}$) & $14$~W \\
Pulses per CPI ($L$)            & $75$  & Eve SINR threshold ($\gamma_{\rm th}$) & $0.1$ \\
SI coefficient ($\rho$)         & $0.1$ & SIC factor ($\lambda_{\rm SIC}$) & $1$ \\
BS position ($\mathbf d_b$)     & $[0,0,0]$~m & UE$_1$ weight ($a_1$)      & $1$ \\
Target position ($\mathbf d_e$) & $[146.9,0,30]$~m & UE$_2$ weight ($a_2$) & $2$ \\ 
UE$_1$ position ($\mathbf d_1$) & $[30,0,0]$~m & Sensing weight ($a_s$)    & $3$ \\
UE$_2$ position ($\mathbf d_2$) & $[70,0,0]$~m & Max AO iterations ($I_{\max}$)& $10$ \\
Target AoA ($\theta_b$)         & $11.5^\circ$ & Convergence threshold ($\varepsilon$) & $10^{-3}$ \\
Target reflection coeff. ($\beta_0$) & $1$ & & \\
\bottomrule
\end{tabular}
}
\end{table}

Numerical simulations are performed to evaluate the proposed robust and secure NOMA--ISAC beamforming design, with the main parameters summarized in Table~\ref{tab:sim_params}. The BS employs $N=8$ antennas and serves two NOMA users with $M=4$ antennas, each while performing monostatic sensing of the target. The carrier frequency is $f_c=6$~GHz, with $L=75$ pulses per CPI.The BS-to-target distance is $150$~m, while the BS-to-UE distances are $30$~m and $70$~m. The power budgets are $P_U^{\max}=1$~W and $P_s^{\max}=14$~W, respectively, and the Eve SINR threshold is $\gamma_{\rm th}=0.1$. The utility weights are set to $a_1=1$, $a_2=2$, and $a_s=3$.

Three approaches are compared in the subsequent results:
\begin{itemize}
  \item \textbf{Proposed}: Our proposed uncertainty-aware and secure design given in Algorithm~\ref{algo}.
  \item \textbf{Ideal}: The security constraint is included, but the design assumes perfect knowledge of Eve's location ($\sigma^2_{\theta_b}=0$, $\sigma^2_{\tau_b}=0$).
 \item \textbf{Non-secure}: The estimation uncertainty is considered, while the Eve SINR constraint is removed ($\gamma_{\rm th}\to\infty$). This case shows the performance loss due to the security constraint under the same sensing uncertainty.
\end{itemize}

%%%%%%%%%%%%%%%%%%%%%%%%%%%%%%%%%%%%%%%%%%%%%

\begin{figure*}[!t]
    \centering
    \subfloat[Utility vs. $p_s^{\max}$\label{fig:cs_vs_ps}]
    {\includegraphics[width=0.33\textwidth]{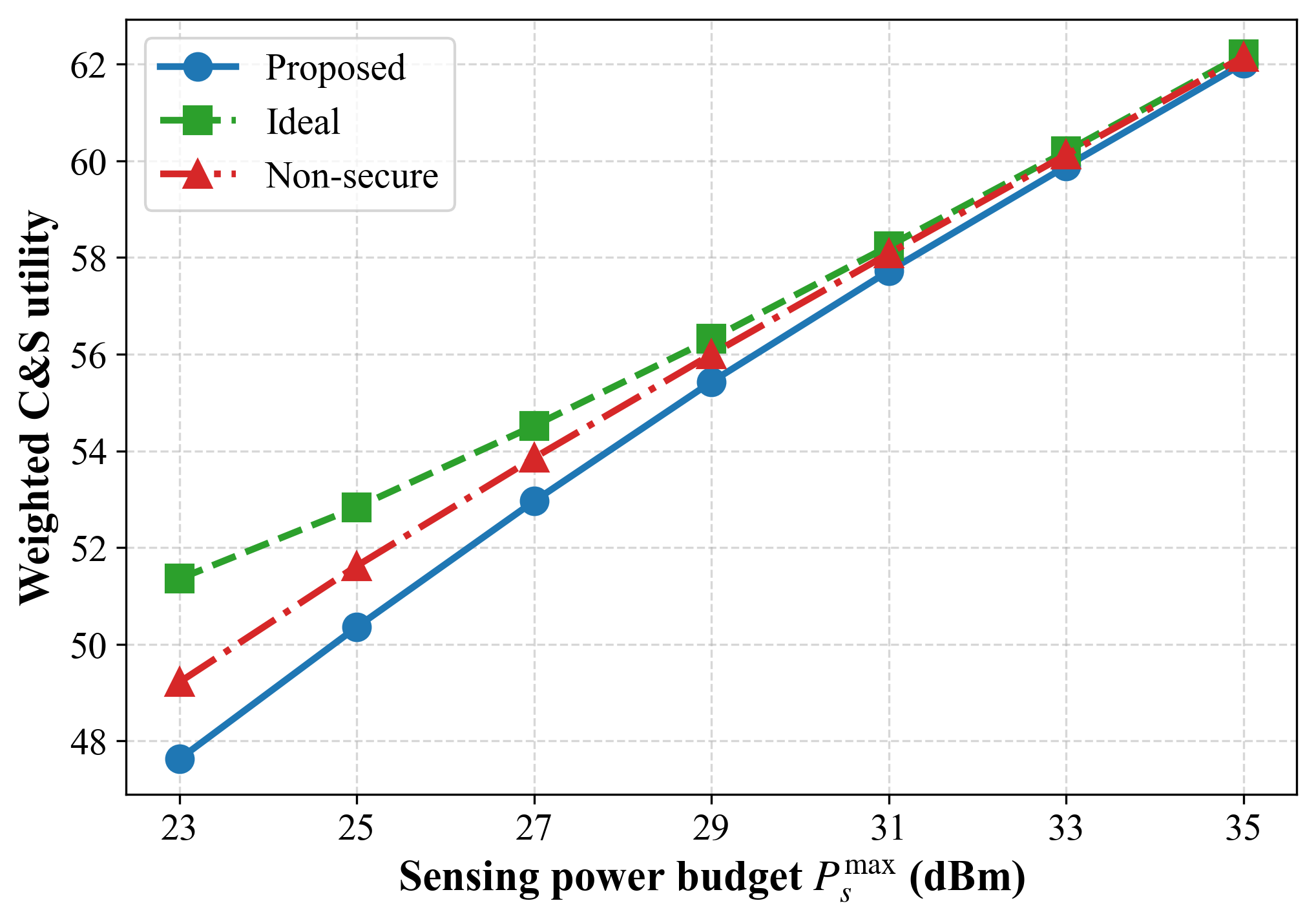}}
    \hfill
    \subfloat[Sensing utility vs. $p_s^{\max}$\label{fig:sensing_vs_ps}]
    {\includegraphics[width=0.33\textwidth]{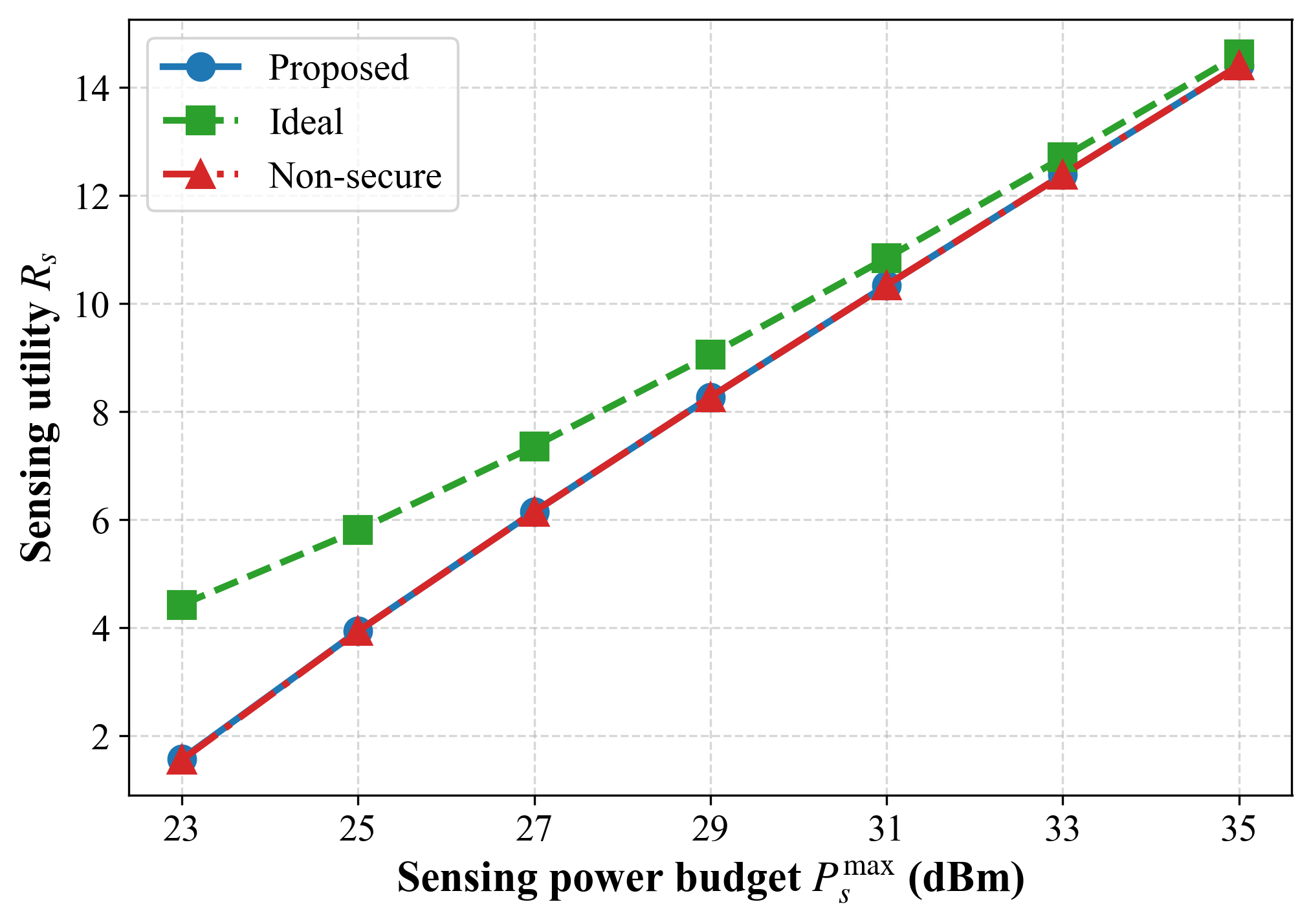}}
    \hfill
    \subfloat[Communication rate vs. $p_s^{\max}$\label{fig:comm_vs_ps}]
    {\includegraphics[width=0.33\textwidth]{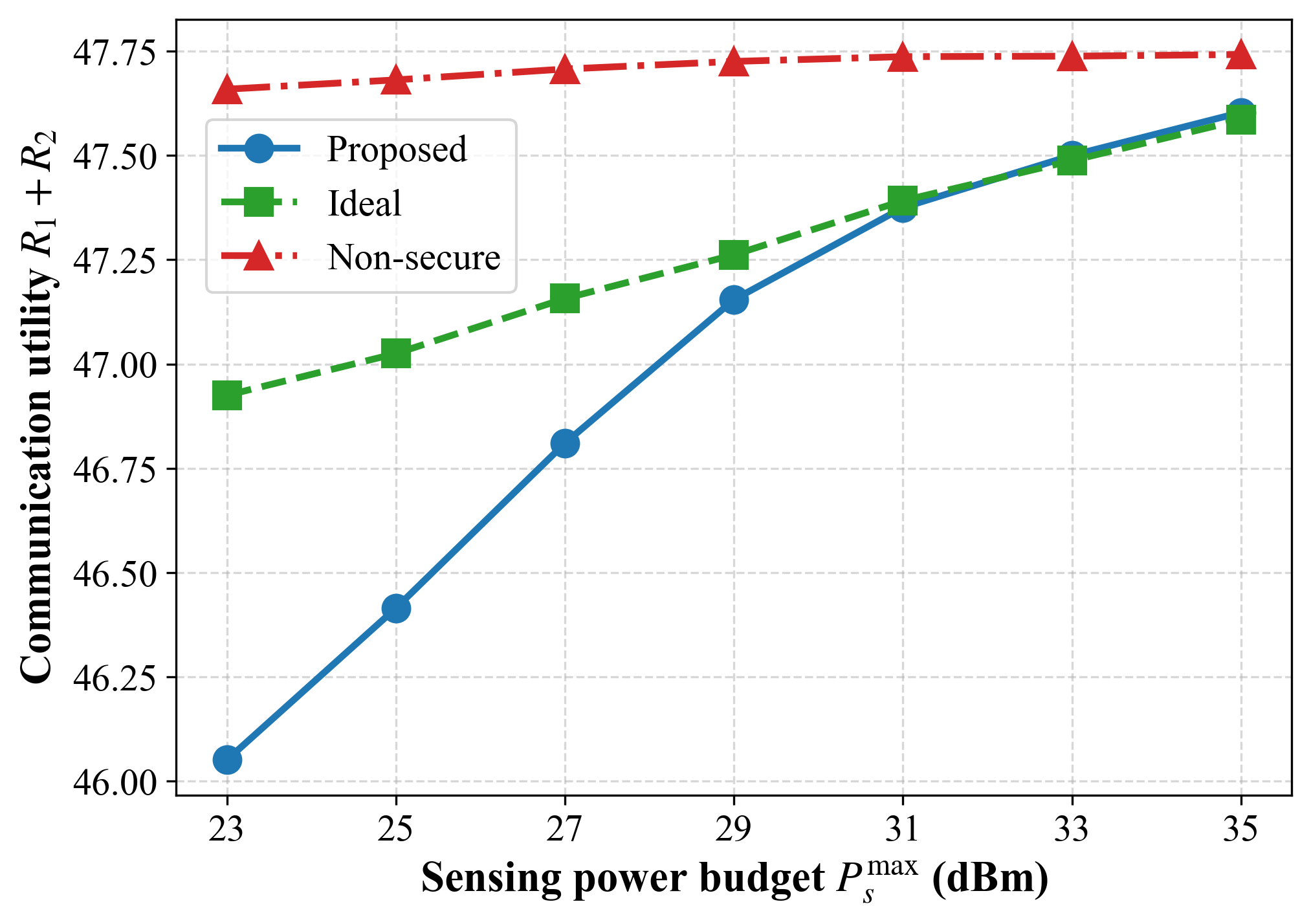}}
    \caption{Performance comparison vs. sensing power budget $P_s^{\max}$.}
    \label{fig:performance_vs_ps_max}
\end{figure*}

Fig.~\ref{fig:performance_vs_ps_max} illustrates the effect of the sensing power budget $P_s^{\max}$ on the three performance metrics. As $P_s^{\max}$ increases, the sensing power constraint in~\eqref{cst3} gives the optimizer more \ac{DoF} to improve the performance. As shown in Fig.~\ref{fig:cs_vs_ps}, the \ac{CS} utility increases with $P_s^{\max}$. This improvement is mainly driven by the sensing component, since a larger $p_s$ strengthens the useful sensing signal in~\eqref{SINR_sensing_final}. Compared with the proposed algorithm, the ideal approach achieves the largest \ac{CS} utility since it assumes perfect knowledge of Eve's location, while the non-secure approach benefits from removing Eve's SINR constraint and thus has more \ac{DoF}. Fig.~\ref{fig:sensing_vs_ps} shows that the sensing utility increases with $P_s^{\max}$ because a larger sensing power strengthens the sensing signal and improves the estimation of Eve's angle and delay, which reduces the uncertainty terms in $\mathbf{Q}_s$ in~\eqref{SINR_sensing_final}. The proposed and non-secure approaches achieve nearly the same sensing utility because Eve's SINR constraint mainly affects the communication beamforming and has little impact on the optimized $p_s$. The ideal approach achieves a higher sensing utility since it assumes perfect knowledge of Eve's location and does not account for location uncertainty in $\mathbf{Q}_s$. However, the gap between the ideal and proposed approaches decreases as $P_s^{\max}$ increases, since the improved estimation accuracy reduces the effect of location uncertainty. Although increasing $P_s^{\max}$ increases the residual SI in~\eqref{eu_SINR_user1_xN} and~\eqref{eu_SINR_user2_xN}, the sensing signal also increases interference at Eve through the term $p_s q_b$. This relaxes Eve's SINR constraint and allows the communication beamformers to be directed more toward the BS, resulting in a higher communication rate. In contrast, the non-secure approach achieves the highest communication rate over the considered range because it does not include Eve's SINR constraint. The proposed approach initially achieves a lower rate than the ideal approach due to the uncertainty in Eve's location. However, the gap gradually decreases as $P_s^{\max}$ increases because the improved estimation accuracy reduces the effect of this uncertainty.

%%%%%%%%%%%%%%%%%%%%%%%%%%%%%%%%%%%%%%%%%%%%%
\begin{figure}
    \centering
    \includegraphics[width=0.68\columnwidth]{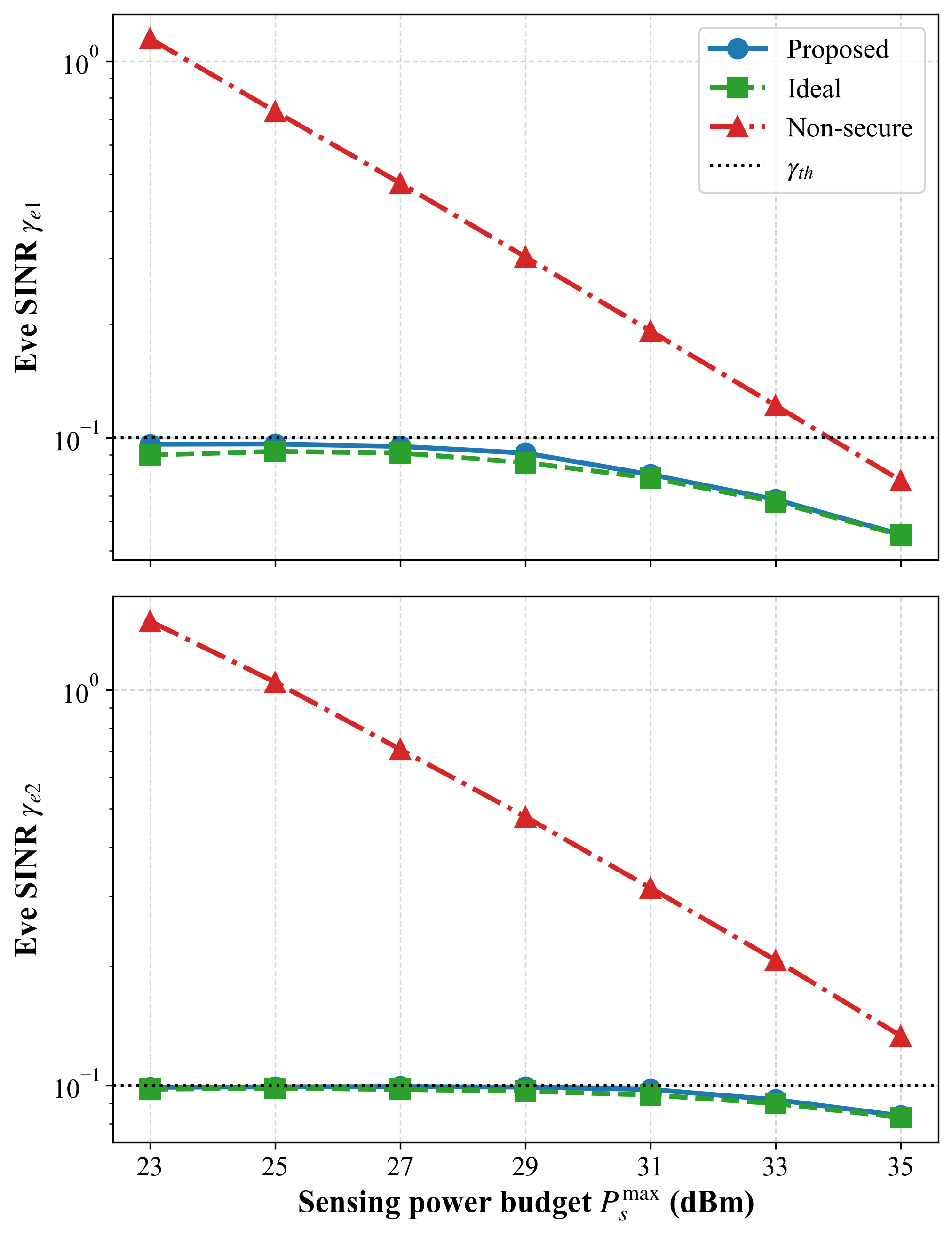}
    \caption{Eve's SINRs vs. $P_s^{\max}$}
    \label{fig:eve_sinr_vs_ps}
\end{figure}

Fig.~\ref{fig:eve_sinr_vs_ps} shows Eve's SINR, $\gamma_{e,1}$ and $\gamma_{e,2}$, vs. $P_s^{\max}$. The dotted horizontal lines represent the secrecy threshold $\gamma_{\mathrm{th}}$. Both the proposed and ideal approaches keep Eve's SINRs below this threshold over all values of $P_s^{\max}$. The ideal approach operates slightly closer to the threshold since it assumes perfect knowledge of Eve's location, whereas the uncertainty terms in~\eqref{eq:Guk} make the proposed design more conservative with slightly lower SINR at Eve. As $P_s^{\max}$ increases, Eve's SINR under both secure approaches decreases. This behavior is consistent with~\eqref{eu_case2_Eve_xNew}, where a larger $p_s$ increases the term $p_s q_b$ in the denominator, and therefore reduces Eve's SINR. The difference between the proposed and ideal approaches also becomes smaller at high sensing powers because the improved estimation of Eve's location reduces the uncertainty considered in the robust design. In contrast, the non-secure approach does not consider Eve's performance and thus allows substantially higher SINRs at Eve. Although both $\gamma_{e1}$ and $\gamma_{e2}$ decrease with $P_s^{\max}$ due to higher sensing interference, $\gamma_{e2}$ remains above the threshold over the considered range, while $\gamma_{e1}$ falls below at the highest sensing power. Thus, increasing $P_s^{\max}$ alone does not consistently satisfy the secrecy requirement, while the proposed approach satisfies the secrecy requirement under location uncertainty with only a small performance gap relative to the ideal case.

%%%%%%%%%%%%%%%%%%%%%%%%%%%%%%%%%%%%%%%%%%%%%

\begin{figure*}[!t]
    \centering
    \subfloat[Secrecy rate vs. $P_s^{\max}$\label{fig:secrecy_vs_ps}]
    {\includegraphics[width=0.38\textwidth]{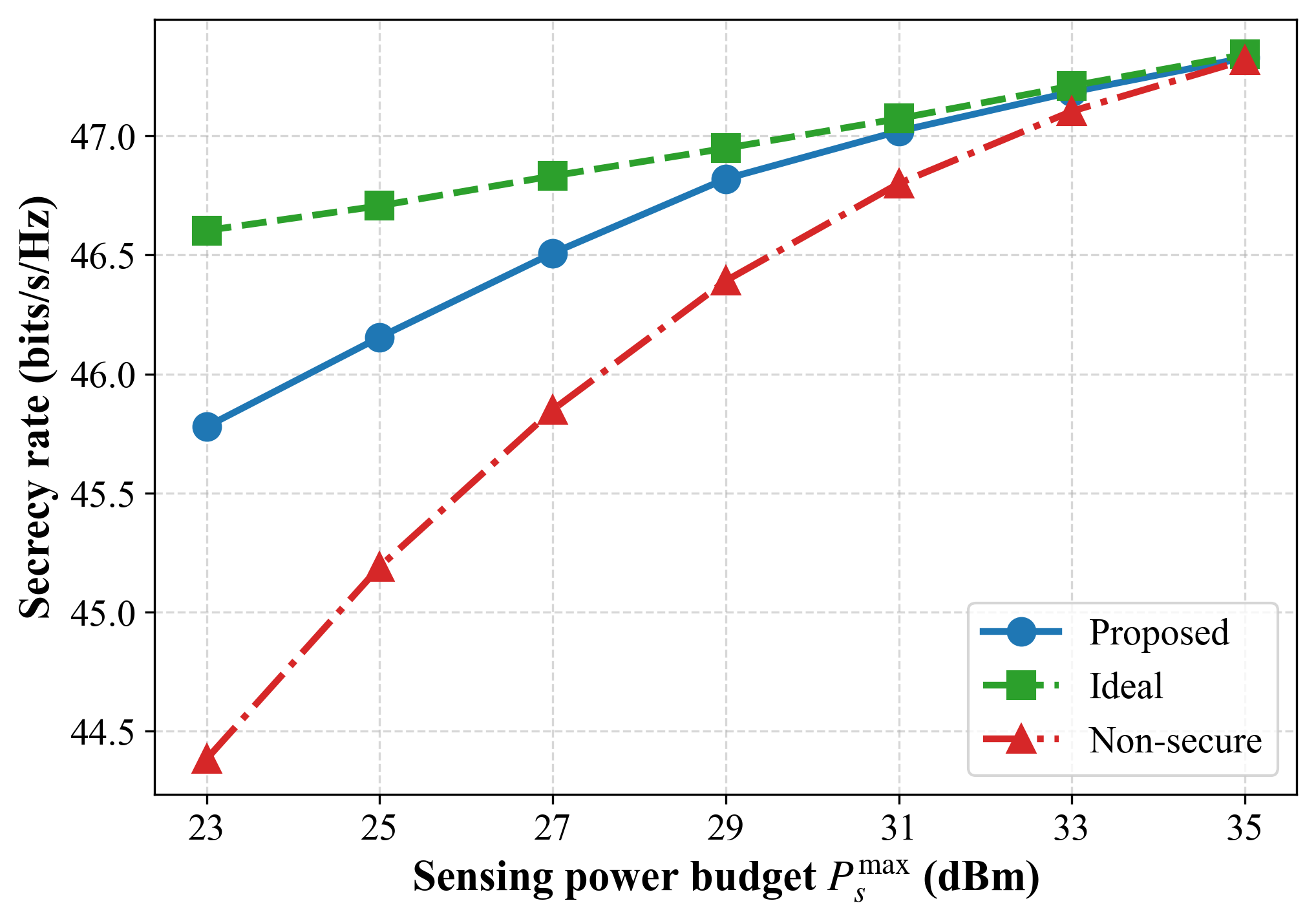}}
    \hspace{1cm}
    \subfloat[Secrecy rate vs. BS--Eve distance\label{fig:secrecy_vs_d}]
    {\includegraphics[width=0.38\textwidth]{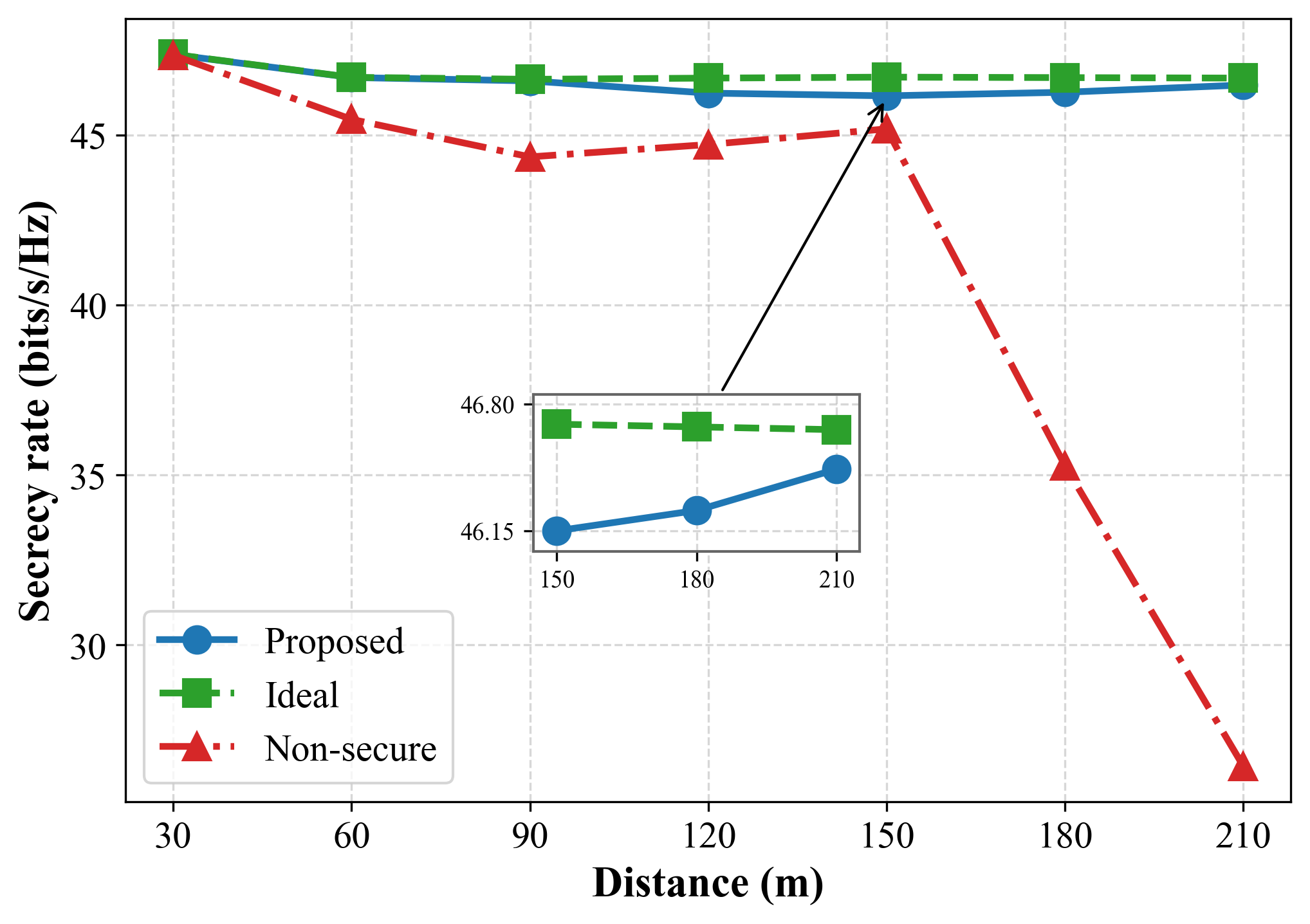}}
    \caption{Secrecy-rate performance comparison: (a) vs. $P_s^{\max}$ and (b) vs. BS--Eve distance.}
    \label{fig:secrecy_performance}
\end{figure*}

Fig.~\ref{fig:secrecy_performance} presents the secrecy-rate performance vs. $P_s^{\max}$ and BS--Eve distance. The weighted sum secrecy rate (WSSR) is defined as 
\[R_{\mathrm{sec}}^{\mathrm{w}}=\sum_{k=1}^{2} a_k\left[\log_2(1+\gamma_k)-\log_2(1+\gamma_{e,k})
\right]^+.\]
As shown in Fig.~\ref{fig:secrecy_vs_ps}, the WSSR increases with $P_s^{\max}$ for all schemes. A larger $p_s$ reduces Eve's effective SINR through the term $p_s q_b$ in~\eqref{eu_case2_Eve_xNew}, which improves the achievable secrecy rate. The proposed approach achieves a lower WSSR than the ideal benchmark at small $P_s^{\max}$ because the uncertainty in Eve's location is considered in the robust design. However, the gap gradually decreases as $P_s^{\max}$ increases, since the improved sensing accuracy reduces the location uncertainty. Although the non-secure approach does not consider Eve's SINR constraint, its WSSR also increases with $P_s^{\max}$ because the sensing signal reduces Eve's SINR.

Fig.~\ref{fig:secrecy_vs_d} shows the effect of the BS--Eve distance on the WSSR. The proposed and ideal approaches achieve a nearly constant WSSR over the considered distance range, showing that the secrecy requirement can be maintained despite the increased uncertainty in Eve's location at larger sensing distances. In contrast, the WSSR of the non-secure approach decreases significantly at large distances, particularly beyond $150$~m. Since this approach does not consider Eve's SINR constraint in the optimization, it cannot provide consistent secrecy performance as the sensing conditions become less accurate. These results show that the proposed approach achieves robust secrecy performance over different sensing powers and BS--Eve distances while remaining close to the ideal benchmark.

% %%%%%%%%%%%%%%%%%%%%%%%%%%%%%%%%%%%%%%%%%%%%%

\section{Conclusion} \label{sec:Conclusion}
We investigated a secure uplink \ac{PD-NOMA} ISAC system in the presence of an eavesdropping target with uncertain location. To account for the sensing uncertainty, a CRB-informed framework was developed to characterize how angle and delay estimation errors propagate into the assumed communication, sensing, and eavesdropping channel models. Based on this uncertainty-aware model, a joint beamforming and sensing-power optimization problem was formulated to maximize a weighted communication--sensing utility subject to secrecy constraints, and an \ac{AO-SCA} algorithm was proposed to obtain an efficient solution. 
Simulation results confirm that the proposed robust-secure design reliably keeps the adopted uncertainty-aware surrogate for Eve's SINR below the prescribed leakage cap, while maintaining strong sensing performance and competitive communication rates. These results demonstrate that CRB-aware robust design is an effective and practically relevant approach for secure uplink NOMA--ISAC, and that ignoring uncertainty in the sensing estimates can lead to overly optimistic and potentially insecure system designs. A natural direction for future work is to extend the framework to mobile eavesdroppers and time-varying uncertainty models.

%%%%%%%%%%%%%%%%%%%%%%%%%%%%%%%%%%%%%%%%%%%%%%%%%%%%%%%%%%%%%
%------------------------  Appendices ----------------------%
%%%%%%%%%%%%%%%%%%%%%%%%%%%%%%%%%%%%%%%%%%%%%%%%%%%%%%%%%%%%%

\appendices
\allowdisplaybreaks[4] % re-enable display breaks in the appendices (was switched off before Section VI)
\numberwithin{equation}{section}
\makeatletter 
\newcommand{\section@cntformat}{Appendix \thesection:\ }
\makeatother

\section{Proof of Lemma \ref{lemma1}}\label{AppendixA}

Substituting \eqref{eu_radar:sig} into \eqref{eu_ys}, the sensing echo is rewritten as
\begin{align}
\ybf_\text{s}(t)=\sqrt{p_s}\,\beta_0\Abf(\theta_b)\abf(\hat{\theta}_b)q(t-2\tau_b).
\label{eu:sensing}
\end{align}
Assuming the small delay error conditions $|\Delta\tau_b|\ll 1/\beta_{\rm rms}$ and $2\Delta\tau_b = 2(\tau_b-\hat{\tau}_b) \ll T$, we write $\ybf_\text{s}$ as
\begin{align}
\widetilde{\ybf}_s &= \frac{1}{\sqrt L}\int_0^{LT}\ybf_s(t)q^*(t-2\hat{\tau}_b)\,dt \nonumber\\
& =\sqrt{\frac{p_s}{L}}\,\beta_0 \Abf(\theta_b)\abf(\hat{\theta}_b)\int_0^{LT}
q(t-2\tau_b)q^*(t-2\hat{\tau}_b)\,dt \nonumber\\
& \approx \sqrt{\frac{p_s}{L}}\, \beta_0 \Abf(\theta_b)\abf(\hat{\theta}_b)
\sum_{m=0}^{L-1}|a_m|^2 R_\phi(2\Delta\tau_b) \nonumber\\
& = \sqrt{p_sL}\, \beta_0 \Abf(\theta_b)\abf(\hat{\theta}_b) R_\phi(2\Delta\tau_b). \label{eu:tilde_ys}
\end{align}
The third term neglects the cross-pulse terms $\sum_{\ell\neq m} a_\ell a_m^* R_\phi\!\left((\ell-m)T+2\Delta\tau_b\right)$ which are small under $|2\Delta\tau_b|\ll T$ since $\phi(t)$ satisfies the Nyquist condition
$R_\phi(nT)=\delta[n]$.

Using the second-order Taylor expansion of $R_\phi(\tau)$ around
$\tau=0$, we have
\begin{align}  \label{eu:Taylor}
&R_\phi (2\Delta\tau_b) \approx \underbrace{R_\phi(0)}_{=1}+ \underbrace{\frac{\partial R_\phi}{\partial\tau}{\bigg\rvert_{\tau=0}}2\Delta\tau_b}_{=0} + \frac{\partial^2R_\phi}{\partial\tau^2}{\bigg\rvert_{\tau=0}}\frac{(2\Delta\tau_b)^2}{2}.
\end{align}
For a unit-energy pulse with symmetric spectrum $|\Phi(f)|^2$,
\begin{align}
R_\phi(0)&=1, \quad R'_\phi(0) = j2\pi\int_{-\infty}^{\infty} f|\Phi(f)|^2df =0,\nonumber\\
R''_\phi(0) &= -(2\pi)^2 \int_{-\infty}^{\infty} f^2|\Phi(f)|^2df = -(2\pi\beta_{\rm rms})^2,
\end{align}
which implies $R'''_\phi(0)=0$, so higher-order terms are neglected. Therefore, substituting \eqref{eu:Taylor} into \eqref{eu:tilde_ys} gives
\begin{align}
\widetilde{\ybf}_\text{s} \approx \sqrt{p_s L}\,
\beta_0 \Abf(\theta_b)\abf(\hat{\theta}_b) \left(1+2R''_\phi(0)\Delta\tau_b^2\right).
\label{eu:y_s_MF}
\end{align}

Similarly, the SI signal in \eqref{eu_SI} can be expressed as
\begin{align}
\ybf_{\rm SI}(t) = \sqrt{p_s}\, \Hbf_{\rm SI}\abf(\hat{\theta}_b)q(t).
\label{eu:SI}
\end{align}
Applying a \ac{MF} on \eqref{eu:SI} over the $L$ samples gives
\begin{align}
\widetilde{\ybf}_{\rm SI} &= \frac{1}{\sqrt L} \int_0^{LT}
\ybf_{\rm SI}(t) q^*(t-2\hat{\tau}_b)\,dt \nonumber\\
&= \sqrt{\frac{p_s}{L}}\, \Hbf_{\rm SI}\abf(\hat{\theta}_b)
\int_0^{LT} q(t)q^*(t-2\hat{\tau}_b)\,dt \nonumber\\
&= \sqrt{\frac{p_s}{L}}\, \Hbf_{\rm SI}\abf(\hat{\theta}_b) r_q(2\hat{\tau}_b),
\label{eu:SI_MF_rq}
\end{align}
where
\begin{equation}
r_q(2\hat{\tau}_b) = \sum_{\ell=0}^{L-1} \sum_{m=0}^{L-1}
a_\ell a_m^*R_\phi\!\left((\ell-m)T+2\hat{\tau}_b\right). \label{eu:rq}
\end{equation}
To express \eqref{eu:rq} in terms of the modulation sequence, we approximate the estimated round-trip delay by its nearest symbol lag, $2\hat{\tau}_b \approx \ell'(\hat{\tau}_b)T,$
and $\ell'(\hat{\tau}_b) = \operatorname{round}\!\left(\frac{2\hat{\tau}_b}{T}\right)$.
%\label{eu:delay_lag}
Using the Nyquist condition, \eqref{eu:rq} is then approximated as
\begin{align}
r_q(2\hat{\tau}_b) &\approx \sum_{\ell=0}^{L-1}
a_\ell a^*_{(\ell+\ell'(\hat{\tau}_b))\bmod L} = c\!\left(\ell'(\hat{\tau}_b)\right),\label{eu:rq_c}
\end{align}
where the circular indexing follows from the periodic waveform model. Therefore,
\begin{align}
\widetilde{\ybf}_{\rm SI} \approx \sqrt{\frac{p_s}{L}}\, \Hbf_{\rm SI}\abf(\hat{\theta}_b)
c\!\left(\ell'(\hat{\tau}_b)\right). \label{eu:y_SI_MF}
\end{align}

Finally, using
$\mathbb E\{\nbf_b(t_1)\nbf_b^H(t_2)\} =\sigma_b^2\Ibf\delta(t_1-t_2)$, the covariance of the filtered noise is
\begin{align}
\mathbb E\!\left\{\widetilde{\mathbf n}_b\widetilde{\mathbf n}_b^H\right\}
&=\frac{1}{L}\int_0^{LT}\!\!\int_0^{LT}\mathbb E\!\left\{
\mathbf n_b(t_1)\mathbf n_b^H(t_2)\right\}\nonumber\\
&\quad\times q^*(t_1-2\hat{\tau}_b) q(t_2-2\hat{\tau}_b)\,dt_1dt_2 \nonumber\\
&= \frac{\sigma_b^2}{L}\mathbf I \int_0^{LT} \left|q(t-2\hat{\tau}_b)\right|^2dt = \sigma_b^2\mathbf I, \label{filterednoise}
\end{align}
where the last equality follows from the periodicity and normalization of $q(t)$, for which $\int_0^{LT}|q(t-2\hat{\tau}_b)|^2dt=L$.

\section{Proof of Lemma \ref{lemma2}}\label{AppendixB}
We need to calculate the three expectations in \eqref{gammas}. As such, using Lemma \ref{lemma1}, we first calculate 
\begin{align} 
    &\mathbb{E}\{ \vert \vbf_s^H \widetilde{\ybf}_\text{s} \vert^2\}
    = \mathbb{E}\{ \vert \beta_0 \vbf_s^H \sqrt{p_sL}
    \Abf(\theta_b)\abf(\hat\theta_b) Q_\phi \vert^2 \} \nonumber\\
    &\stackrel{(a)}{\approx} \mathbb{E}\Big\{ 
    \underbrace{\Big\vert \beta_0 \vbf_s^H \sqrt{p_sL}
    \Abf(\theta_b)
    \Big[\abf(\theta_b)
    - \abf'_{\theta_b}\Delta\theta_b
    + \frac{1}{2}\abf''_{\theta_b}\Delta\theta_b^2 \Big]
    Q_\phi \Big\vert^2}_{=I} \Big\},
    \label{eu:expected_s}
\end{align}
where $Q_\phi \triangleq 1+2R''_\phi(0)(\Delta\tau_b)^2$, $\abf'_{\theta_b}\triangleq 
\frac{\partial\abf(\theta_b)}{\partial{\theta_b}}{\big\rvert_{\theta=\theta_b}}$, 
$\abf''_{\theta_b}\triangleq \frac{\partial^2\abf(\theta_b)}{\partial{\theta_b^2}}{\big\rvert_{\theta=\theta_b}}$, and $(a)$ follows from the second-order Taylor approximation of $\abf(\hat\theta_b)$ around $\theta_b$, with $\Delta\theta_b \triangleq\theta_b-\hat\theta_b$. 
Expanding the squared magnitude in \eqref{eu:expected_s}, we obtain
\begin{align}
\begin{split}
I &\approx \vert \beta_0 \vert^2 p_s L \vbf_s^H \Big(\Big[\Abf(\theta_b)\abf(\theta_b)
+ \Abf(\theta_b)\abf(\theta_b)2R''_\phi(0)\Delta\tau_b^2\\
&\qquad 
- \Abf(\theta_b)\abf'_{\theta_b}\Delta\theta_b
+ \frac{1}{2}\Abf(\theta_b)\abf''_{\theta_b}\Delta\theta_b^2\Big]\\
&\qquad\times \Big[\abf(\theta_b)^H\Abf(\theta_b)^H
+ \abf(\theta_b)^H\Abf(\theta_b)^H2R''_\phi(0)\Delta\tau_b^2\\
&\qquad
- \abf_{\theta_b}'^{H}\Abf(\theta_b)^H\Delta\theta_b
+ \frac{1}{2}\abf_{\theta_b}''^{H}\Abf(\theta_b)^H\Delta\theta_b^2\Big]\Big)\vbf_s .
\end{split}
\end{align}
Taking the expectation, retaining terms up to second order in the angle and delay errors, and neglecting products whose total order exceeds two, we obtain
\begin{align}
\mathbb{E}\!\left\{\left|\mathbf v_s^H\widetilde{\mathbf y}_s\right|^2\right\}
\approx |\beta_0|^2p_sL \mathbf v_s^H\mathbf Q_s\mathbf v_s.
\label{eu_final_sensing}
\end{align}
Here, $\mathbf Q_s$ is defined in\eqref{eu_expected_c}, and the errors are modeled as zero-mean with
$\mathbb E\{\Delta\theta_b^2\}=\sigma_{\theta_b}^2$ and $\mathbb E\{\Delta\tau_b^2\}=\sigma_{\tau_b}^2$.

Next we calculate $\mathbb{E}\{ \vert \vbf_s^H \widetilde\ybf_{\rm SI} \vert^2 \}$ using Lemma \ref{lemma1} as
\begin{align} \nonumber
    \mathbb{E}\{ \vert \vbf_s^H \widetilde\ybf_{\rm SI} \vert^2 \} &= \mathbb{E}\{ \vert \vbf_s^H \sqrt{\frac{p_s}{L}}\, r_q(2\hat\tau_b) \, \Hbf_{\rm SI}\abf(\hat\theta_b) \vert^2 \} \\ \label{den1}
    &\approx \frac{p_s}{L} |c(\ell'(\hat\tau_b))|^2 \, ( \vbf_s^H \Qbf_{\rm SI} \vbf_s).
\end{align}

Finally, the filtered-noise covariance was derived in Lemma~\ref{lemma1}. Thus, substituting \eqref{eu_final_sensing}, \eqref{den1}, and \eqref{filterednoise} into \eqref{gammas} leads to the sensing SINR expression in Lemma~\ref {lemma2}, which completes the proof.

\section{Proof of Eve's SINR} \label{AppendixD}
\subsection{Calculation of actual channel between BS and Eve} \label{A:BS_Eve}
We define the true channel from BS to Eve in \eqref{eu_truechannel_b}, where the steering vector is 
\begin{align} \label{eu_theta_b}
\abf(\theta_b) &= \frac{1}{\sqrt N}\left[ 1 \; \; \exp(j2\pi\delta \sin(\theta_b)/\lambda) \; \cdots \right. \nonumber\\ 
 & \qquad \left. \exp(j2\pi\delta (N-1)\sin(\theta_b)/\lambda) \right]^T \;.
\end{align}
The derivative of $\widetilde\gbf_{b}$ w.r.t. $\theta_{b}$ based on \eqref{eu_theta_b} is
\begin{align} \nonumber
     \frac{\partial \abf(\theta_b) }{\partial \theta_b }&= \frac{1}{\sqrt N}\big[0, j\frac{2\pi \delta}{\lambda} \cos\theta_b e^{j\frac{2\pi \delta}{\lambda}\sin\theta_b},\dots,\\ \label{eu_partial_Theta_b}
   &  {j\frac{2\pi \delta}{\lambda}}(N-1)\cos{\theta_b}e^{j \frac{2\pi \delta}{\lambda}{(N-1) \sin \theta_b}}\big]
\end{align}
The derivative of $\widetilde{\mathbf{g}}_{b}$ with respect to $\tau_b$ based on \eqref{eu:alpha_b} is
\begin{align} \label{eu_partial_tau_b}
\frac{\partial\alpha_1(\tau_b)}{\partial\tau_b}
= -\frac{\lambda}{4\pi c \tau_b^2}
= -\frac{c\lambda}{4\pi d_b^2},
\end{align}
Let $\frac{\partial \alpha_1(\tau_{b})}{\partial {\tau_{b}}}=\alpha'_{\tau_b}$ and $\frac{\partial \abf(\theta_b) }{\partial \theta_b }= \abf'_{\theta_b}$.

Now we can calculate  
\begin{align} \nonumber
    & \mathbb{E}\{ \vert\widetilde{\gbf}_{b}^H \xbf \vert^2 \} = \mathbb{E}\{\vert \sqrt{p_s}
    \alpha_1(\hat\tau_b) \abf(\hat\theta_b)^H \abf(\hat\theta_b) \, q(t) \\ \nonumber 
    & \qquad\qquad
    + \sqrt{p_s}\, \frac{\partial\alpha_1(\tau_b)}{\partial{\tau_b}}{\bigg\rvert_{\tau=\hat\tau_b}}\Delta\tau_b \abf(\hat\theta_b)^H \abf(\hat\theta_b) \, q(t) \\ \label{eu_expected_g1}  & \qquad \qquad +\sqrt{p_s}\, \alpha_1(\hat\tau_b)\abf(\hat\theta_b)^H\frac{\partial\abf(\theta_b)}
    {\partial{\theta_b}}{\bigg\rvert_{\theta=\hat\theta_b}} \Delta\theta_b \, q(t) \vert^2 \} 
\end{align}
Expanding the squared magnitude in \eqref{eu_expected_g1}, we obtain 
\begin{align} \nonumber
    & p_s  \big( \alpha_1(\hat\tau_b)^2 \|\abf(\hat\theta_b)\|^4 + \sigma^2_{\tau_b} \alpha_{\tau_b}'^2 \|\abf(\hat\theta_b)\|^4 \\ \nonumber
    & \quad + \sigma_{(\tau\theta)_b} \alpha_{\tau_b}' \alpha_1(\hat\tau_b)^H \abf_{\theta_b}'^H \abf(\hat\theta_b)   \|\abf(\hat\theta_b)\|^2 \\ \nonumber
    & \quad + \sigma_{(\theta\tau)_b} \alpha_{\tau_b}'^H \alpha_1(\hat\tau_b) \abf(\hat\theta_b)^H \abf_{\theta_b}' \|\abf(\hat\theta_b)\|^2  \\ \label{eu_final_g_b}
    & \quad + \sigma^2_{\theta_b} \alpha_1(\hat\tau_b)^2 \vert \abf(\hat\theta_b)^H \abf'_{\theta_b} \vert^2  \big)
\end{align}
Hence, $\mathbb{E}\{ \vert\widetilde{\gbf}_{b}^H \xbf \vert^2 \} = p_s q_b$, where $q_b$ is defined in \eqref{eu_qb}.

\subsection{Calculation of actual channel between users and Eve} \label{B:UE_Eve}
The true channel from UE$_k$ to Eve is defined in \eqref{eu:truechannel_uk}, where $\abf(\theta_{u_k})$ is the array response vector for user $k$:
\begin{align} \nonumber
   \abf(\theta_{u_k}) &= \frac{1}{\sqrt M}\left[ 1 \; \; \exp(j2\pi\delta\sin(\theta_{u_k})/\lambda) \; \cdots \right. \\ & \qquad \left. \label{eu_theta_u}
   \exp(j2\pi\delta (M-1)\sin(\theta_{u_k})/\lambda) \right]^T \;.
\end{align}
Based on \eqref{eu_theta_u} we have:
\begin{align} \nonumber
     \frac{\partial \abf(\theta_{u_k}) }{\partial \theta_{u_k}} &= \frac{1}{\sqrt M} \big[0, j\frac{2\pi \delta}{\lambda} \cos\theta_{u_k} e^{j\frac{2\pi \delta}{\lambda}\sin\theta_{u_k}},\dots,\\ \label{eu_partial_theta_u}
   &  {j\frac{2\pi \delta}{\lambda}}(M-1)\cos{\theta_{u_k}}e^{j \frac{2\pi \delta}{\lambda}{(M-1) \sin \theta_{u_k}}}\big] ,
\end{align}
where $\theta_{u_k} = \arcsin \frac{z_e}{\|\dbf_k - \dbf_e\|}$, and $z_e$ is Eve's altitude.
 
The derivative of $\gbf_{u_k}$ with respect to $\tau_{u_k}$ based on \eqref{eu:alpha_u} is
\begin{align} \label{eu_partial_tau_u}
   \frac{\partial\alpha_1(\tau_{u_k})}{\partial\tau_{u_k}}
   = -\frac{\lambda}{4\pi c \tau_{u_k}^2}
   = -\frac{c\lambda}{4\pi d_{u_k}^2}.
\end{align}
Let $\frac{\partial \alpha_1(\tau_{u_k})}{\partial {\tau_{u_k}}}=\alpha'_{\tau_u}$ and $\frac{\partial \abf(\theta_{u_k}) }{\partial \theta_{u_k} }= \abf'_{\theta_{u_k}}$.
Using the UE$_k$--Eve covariance obtained in Lemma~\ref{lemma3} and the first-order channel approximation in
\eqref{gk_actual}, the average received power from UE$_k$ at Eve is
\begin{align} \nonumber
    \mathbb{E}\{ \vert  \widetilde{\gbf}_{u_k}^H \widetilde\wbf_k \vert^2 \}& = \widetilde\wbf_k^H \mathbb{E}\{\widetilde{\gbf}_{u_k} \widetilde{\gbf}_{u_k}^H \} \widetilde\wbf_k \\ \nonumber
    & =  \alpha_1(\hat\tau_{u_k})^2 \widetilde\wbf_k^H \abf(\hat\theta_{u_k}) \abf(\hat\theta_{u_k})^H \widetilde\wbf_k \\ \nonumber
    &\quad +\sigma^2_{\theta_{u_k}} \alpha_1(\hat\tau_{u_k})^2 \widetilde\wbf_k^H \abf'_{\theta_{u_k}} \abf'^{H}_{\theta_{u_k}}  \widetilde\wbf_k \\ \nonumber
   & \quad + \sigma_{(\tau\theta)_{u_k}}\alpha_1(\hat\tau_{u_k}) \alpha'_{\tau_u} \widetilde\wbf_k^H \abf(\hat\theta_{u_k})\abf'^H_{\theta_{u_k}} \widetilde\wbf_k \\ \nonumber
   &\quad+ \sigma_{(\theta\tau)_{u_k}}\alpha'_{\tau_u} \alpha_1(\hat\tau_{u_k})\widetilde\wbf_k^H \abf'_{\theta_{u_k}} \abf(\hat\theta_{u_k})^H\widetilde\wbf_k\\ 
   &\quad  + \sigma^2_{\tau_{u_k}} \alpha'^2_{\tau_u} \widetilde\wbf_k^H  \abf(\hat\theta_{u_k}) \abf(\hat\theta_{u_k})^H\widetilde\wbf_k.
\end{align}

\ifCLASSOPTIONcaptionsoff
  \newpage
\fi

%%%%%%%%%%%%%%%%%%%%%%%%%%%%%%%%%%%%%%%%%%%%%%%%%%%%%%%%%%%%%
%------------------------  References ----------------------%
%%%%%%%%%%%%%%%%%%%%%%%%%%%%%%%%%%%%%%%%%%%%%%%%%%%%%%%%%%%%%

% %%%%%%%%%%%%%%%%%%%%%%%%%%%%%%%%%%%%%%%%%%%%%%%%%%%%%%%%%%%%%
% %------------------------  Author Bio ----------------------%
% %%%%%%%%%%%%%%%%%%%%%%%%%%%%%%%%%%%%%%%%%%%%%%%%%%%%%%%%%%%%%

\begin{IEEEbiography}[{\includegraphics[width=1in,height=1.25in,clip,keepaspectratio]{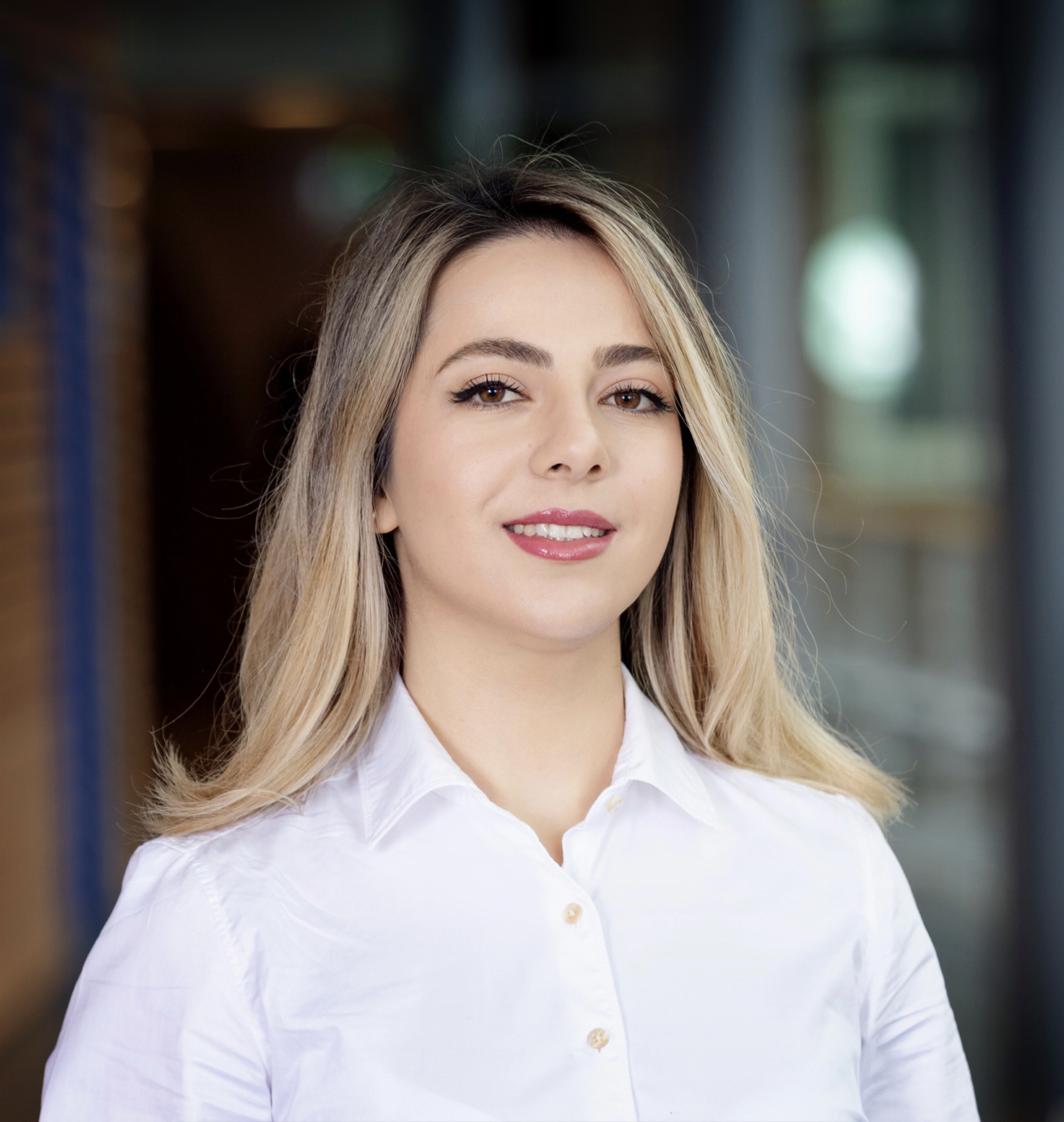}
}]
{Azadeh Tabeshnezhad} (Student Member, IEEE) 
received her M.Sc. degree in Communication Systems from the Science and Research Branch, Tehran, Iran, in 2018. She is currently pursuing a Ph.D. degree at Chalmers University of Technology, Gothenburg, Sweden. In 2022, 2024, and 2025, she was a visiting researcher at the University of California, Irvine, USA, under the co-supervision of Prof. Lee Swindlehurst. Her research interests include non-orthogonal multiple access, (Non-)convex optimization, reconfigurable intelligent surfaces, and integrated sensing and communications.
\end{IEEEbiography}

\begin{IEEEbiography}[{\includegraphics[width=1in,height=1.25in,clip,keepaspectratio]{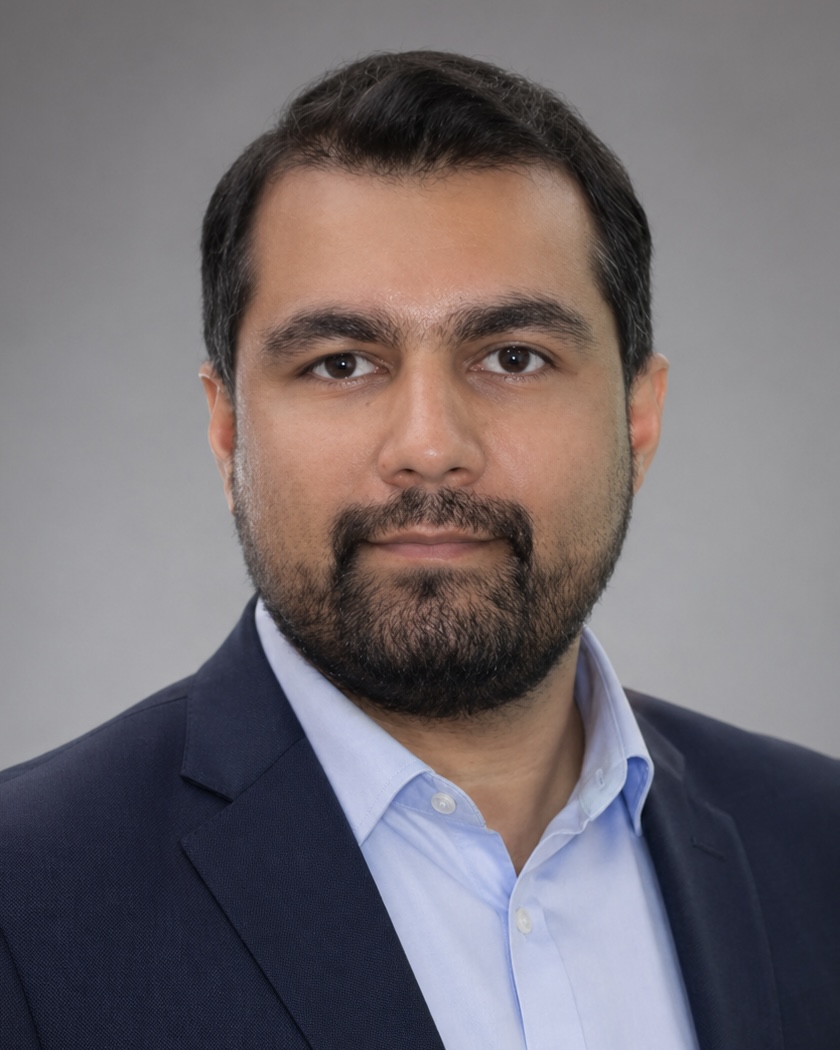}
}]{Milad Tatar Mamaghani} is a researcher based in Sydney, Australia. He received the B.Sc. degree in Electrical Engineering from Amirkabir University of Technology (Tehran Polytechnic) in 2016 and the Ph.D. degree in Electrical Engineering from Monash University in 2022. From 2023 to 2025, he was a Postdoctoral Research Fellow at the Australian National University (ANU). Milad has authored numerous high-impact publications in leading IEEE journals and conferences. He has served as a Symposium Session Chair for IEEE GLOBECOM 2023, a Technical Program Committee (TPC) member for IEEE VTC 2024, and a reviewer for various IEEE transactions and flagship conferences. He received the Best Paper Award at the 2024 IEEE International Conference on Communications (ICC) and the 2024 IEEE Communications Society SPCC Technical Committee Best Paper Award. His research interests include 5G-and-beyond wireless communications, integrated communications and sensing, cyber-physical security, convex optimization, and machine learning.
\end{IEEEbiography}

\begin{IEEEbiography}
[{\includegraphics[width=1in,height=1.25in,clip,keepaspectratio]{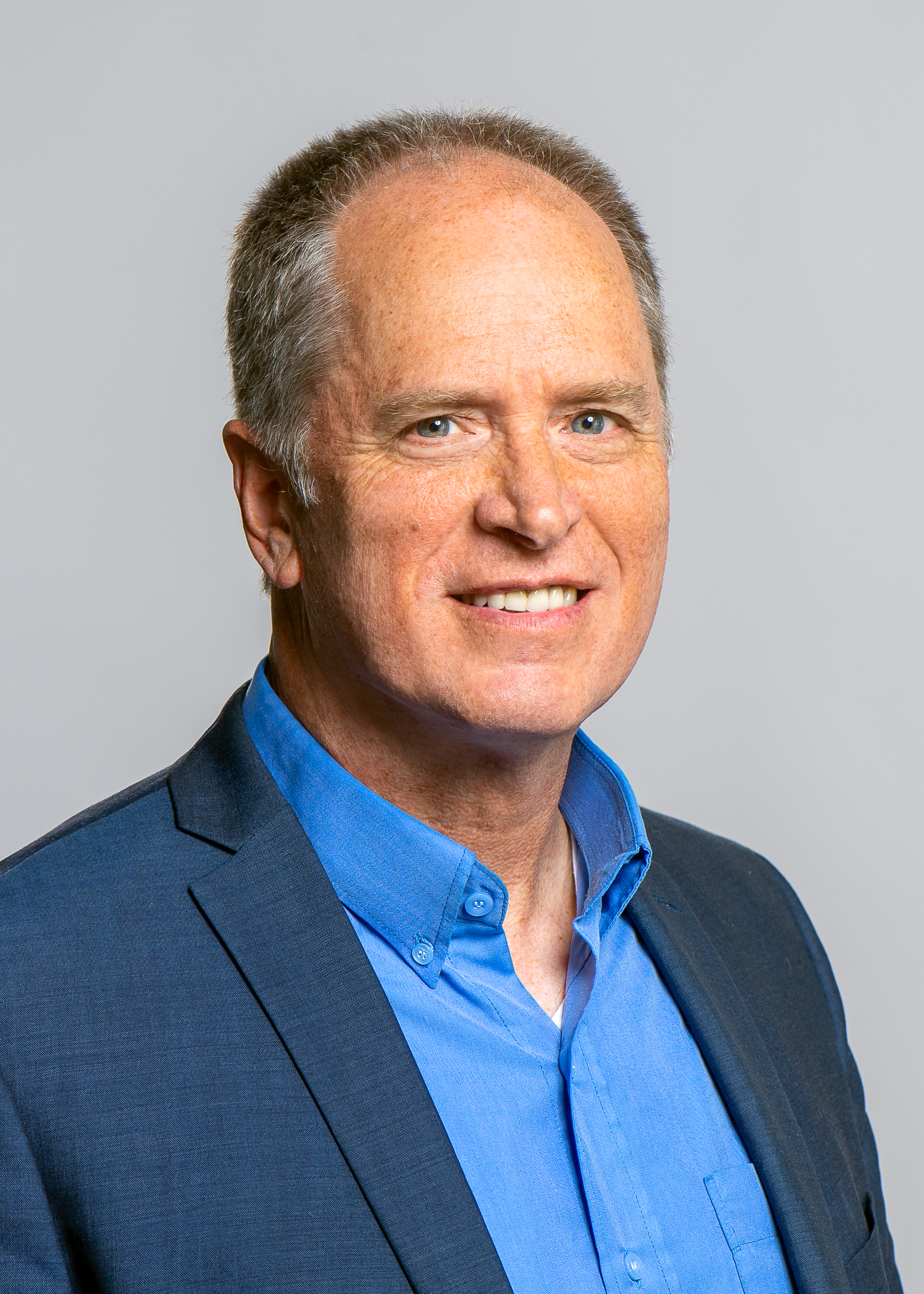}
}]
{A. Lee Swindlehurst} (Fellow, IEEE)
received the B.S. (1985) and M.S. (1986) degrees in Electrical Engineering from Brigham Young University (BYU), and the PhD (1991) degree in Electrical Engineering from Stanford University. He was with the Department of Electrical and Computer Engineering at BYU from 1990-2007, where he served as Department Chair from 2003-06.  During 1996-97, he held a joint appointment as a visiting scholar at Uppsala University and the Royal Institute of Technology in Sweden. From 2006-07, he was on leave working as Vice President of Research for ArrayComm LLC in San Jose, California. Since 2007, he has been with the Electrical Engineering and Computer Science (EECS) Department at the University of California, Irvine, where he is a Distinguished Professor and currently serves as Department Chair. Dr. Swindlehurst is a Fellow of the IEEE. During 2014-17, he was also a Hans Fischer Senior Fellow in the Institute for Advanced Studies at the Technical University of Munich, and in 2016, he was elected as a Foreign Member of the Royal Swedish Academy of Engineering Sciences (IVA). He received the 2000 IEEE W. R. G. Baker Prize Paper Award, the 2006 IEEE Communications Society Stephen O. Rice Prize in the Field of Communication Theory, 2006, 2010 and 2021 IEEE Signal Processing Society’s Best Paper Awards, the 2017 IEEE Signal Processing Society Donald G. Fink Overview Paper Award, a Best Paper award at the 2020 and 2024 IEEE International Conferences on Communications, the 2022 Claude Shannon-Harry Nyquist Technical Achievement Award from the IEEE Signal Processing Society, and the 2024 Fred W. Ellersick Prize from the IEEE Communications Society. His research focuses on array signal processing for radar, wireless communications, and biomedical applications.
\end{IEEEbiography}

\begin{IEEEbiography}[{\includegraphics[width=1in,height=1.25in,clip,keepaspectratio]{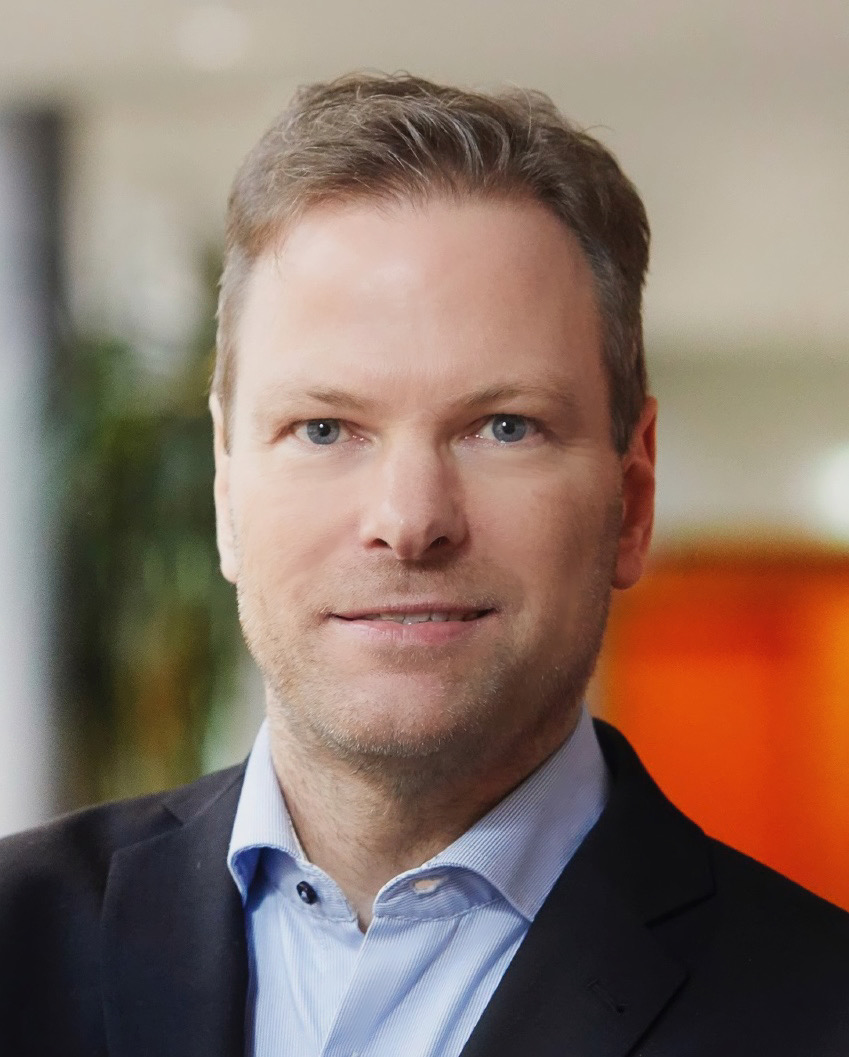}
}]
{Tommy Svensson}(Senior, IEEE) is a Full Professor in Communication Systems at Chalmers University of Technology in Gothenburg, Sweden, where he is leading the Wireless Systems research on air interface and wireless backhaul networking technologies for future wireless systems. He received a Ph.D. in Information theory from Chalmers in 2003, and he has worked at Ericsson AB with core networks, radio access networks, and microwave transmission products. He was involved in the European WINNER I/II/+ and ARTIST4G projects that made important contributions to the 3GPP LTE standards, the EU FP7 METIS and the EU H2020 5GPPP mmMAGIC and 5GCar projects towards 5G, and the Hexa-X-I/II, RISE-6G, SEMANTIC, ROBUST-6G and ECO-eNET projects towards 6G, as well as in the ChaseOn/Bridge Center/emerging WiTECH antenna systems excellence centers at Chalmers targeting mm-wave and (sub)-THz solutions for 5G/6G access, backhaul/ fronthaul and V2X scenarios. His main research interests are in the design and analysis of mobile communication systems, physical layer algorithms, multiple access, resource allocation, cooperative/ situational-aware communications, mm-wave/ sub-THz communications, C-V2X, ISAC, physical-layer security, non-terrestrial networks, sustainable design, and end-to-end architecture. He has co-authored 7 books, 150 journal papers, 180 conference papers, and 80 public EU projects deliverables. He is founding editorial board member and editor of IEEE JSAC Series on Machine Learning in Communications and Networks, has been Chairman of the awards winning IEEE Sweden joint Vehicular Technology/ Communications/ Information Theory Societies chapter, editor of IEEE Transactions on Wireless Communications, IEEE Wireless Communications Letters, Guest editor of several top journals, organized several tutorials and workshops at top IEEE conferences, Lead local organizer of EuCNC \& 6G Summit 2023, and served as coordinator of the Communication Engineering Master’s Program at Chalmers. He is leading the Swedish VR Research Environment “Foundational Algorithms, Protocols, and Systems for Multi-Tier 6G-Non-Terrestrial Networks Integrated Communication and Environmental Sensing (6G-NTN-E)” that started in 2025, and since 2022, he has been a board member of the Swedish telecommunications regulator (PTS).
\end{IEEEbiography}

\begin{IEEEbiography}[{\includegraphics[width=1in,height=1.25in,clip,keepaspectratio]{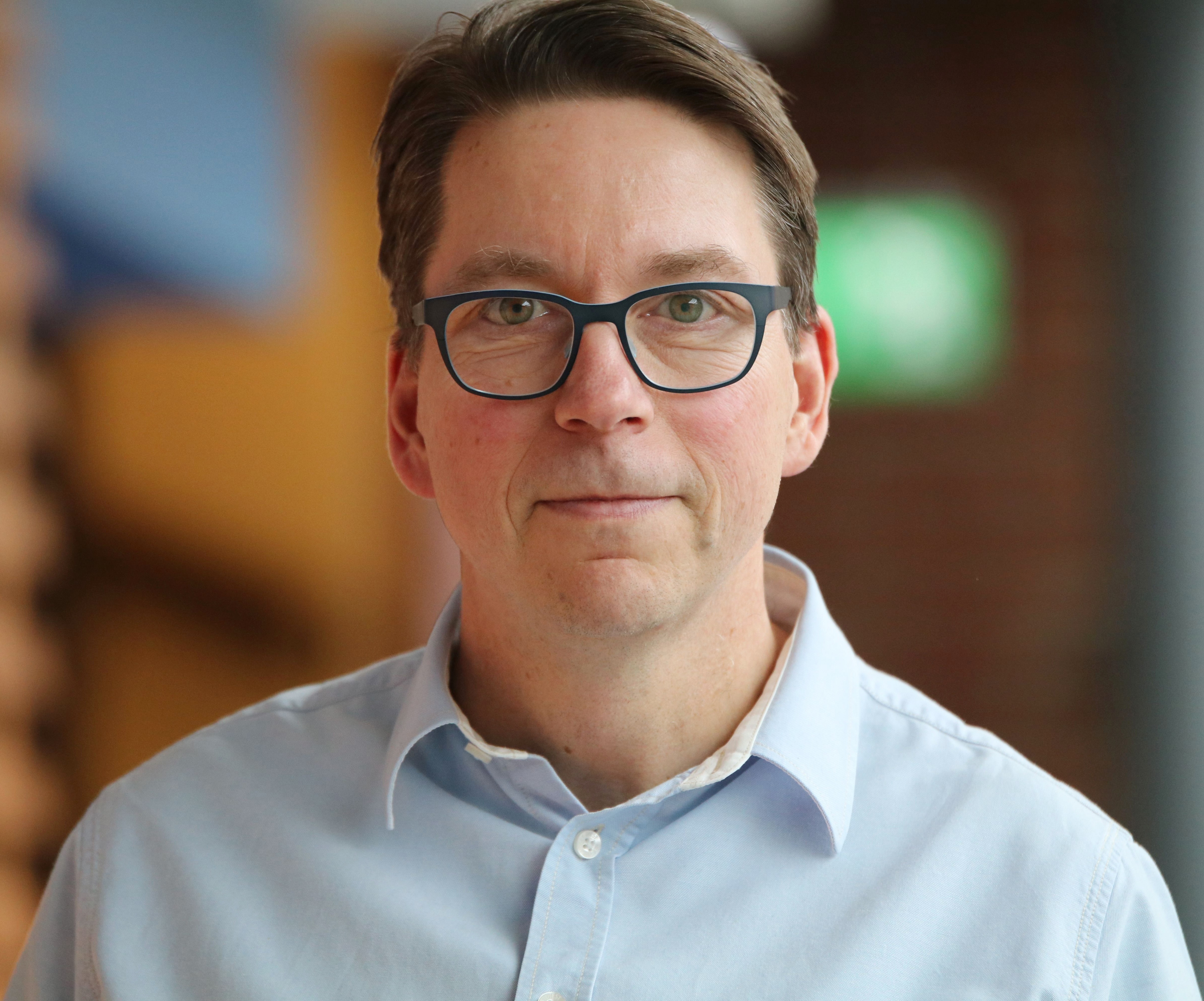}
}]{Erik G. Ström} (Fellow, IEEE) received the M.S. degree in electrical engineering from the Royal Institute of Technology (KTH), Stockholm, Sweden, in 1990, and the Ph.D. degree in electrical engineering from the University of Florida, Gainesville, in 1994. He accepted a post-doctoral position at the Department of Signals, Sensors, and Systems, KTH, in 1995. In February 1996, he was appointed as an Assistant Professor at KTH, and in June 1996, he joined the Chalmers University of Technology, Gothenburg, Sweden, where he has been a Professor of communication systems since June 2003. He currently heads the Division of Communications, Antennas, and Optical Networks and is the Director of Chalmers’ Area-of-Advance Information and Communication Technology. Since 1990, he has acted as a consultant for the Educational Group for Individual Development, Stockholm. His research interests include signal processing and communication theory in general, and constellation labelings, channel estimation, synchronization, multiple access, medium access, multiuser detection, wireless positioning, and vehicular communications in particular. He was a member of the Board of the IEEE VT/COM Swedish Chapter 2000–2006. He received the Chalmers Pedagogical Prize in 1998, the Chalmers Ph.D. Supervisor of the Year Award in 2009, and the Chalmers Area of Advance Award in 2020. He has been a Senior Editor of IEEE TRANSACTIONS ON INTELLIGENT TRANSPORT SYSTEMS, a contributing author, and associate editor for Roy. Admiralty Publishers FesGas-series, and was a Co-Guest Editor for PROCEEDINGS OF THE IEEE special issue on Vehicular Communications in 2011 and IEEE JOURNAL ON SELECTED AREAS IN COMMUNICATIONS special issues on Signal Synchronization in Digital Transmission Systems in 2001 and on Multiuser Detection for Advanced Communication Systems and Networks in 2008.
\end{IEEEbiography}

\end{document}